\documentclass{article}

\usepackage{amsmath} 
 
 \usepackage{comment}
 \usepackage{xspace}
\usepackage{pstricks, pst-node,pst-tree, multido,pst-eps}
\usepackage{lscape}
\usepackage{amsthm, amssymb, amstext}
\usepackage{multicol}
\usepackage{lmodern}
\usepackage{color}
\usepackage{lineno}

\newcommand{\Kf}{red\xspace}
\newcommand{\Pf}{blue\xspace}
\newcommand{\NP}{\ensuremath{\mathbb{NP}}}
\newcommand{\PP}{\ensuremath{\mathbb{P}}}

\newtheorem{lemma}{Lemma}[section]
\newtheorem{proposition}{Proposition}[section]

\newtheorem{theorem}{Theorem}[section]
\newtheorem{remark}{Remark}[section]
\newtheorem{claim}{Claim}
\newtheorem{corollary}{Corollary}[section]

\theoremstyle{definition}

\newcommand{\ol}{\overline}
\newcommand{\sm}{\setminus}
  \newcommand{\cC}{\mathcal{C}}            
  \newcommand{\ie}{\textsl{i.e.}\xspace}   

\author{Faisal N.~Abu-Khzam, Carl Feghali, Haiko M{\"u}ller}

\title{Partitioning a Graph into Disjoint Cliques and a Triangle-free Graph}

\begin{document}
\maketitle

\begin{abstract}
A graph $G = (V, E)$ is \emph{partitionable} if there exists a partition $\{A, B\}$ of $V$ such that $A$ induces a disjoint union of cliques and $B$ induces a triangle-free graph. 
 In this paper we investigate the computational complexity of deciding whether a graph is partitionable. The problem is known to be $\NP$-complete on arbitrary graphs. Here it is proved that if a graph $G$ is bull-free, planar, perfect, $K_4$-free or does not contain certain holes then deciding whether $G$ is partitionable is $\NP$-complete. This answers an open question posed by Thomass{\'e}, Trotignon and Vu\v{s}kovi{\'c}. In contrast a finite list of forbidden induced subgraphs is given for partitionable cographs.  
\end{abstract}

\section{Introduction}\label{introduction}
A graph $G=(V,E)$ is \emph{monopolar} if there exists a partition $\{A, B\}$ of $V$ such that $A$ induces a disjoint union of cliques and $B$ forms an independent set. The class of monopolar graphs has been extensively studied in recent years. It is known that deciding whether a graph is monopolar is $\NP$-complete~\cite{alastair}, even when restricted to triangle-free graphs \cite{churchley1} and planar graphs \cite{vanbang}. In contrast the problem is tractable on several graph classes: a non-exhaustive list includes cographs \cite{ekim}, polar permutation graphs \cite{meister}, chordal graphs \cite{werra}, line graphs \cite{churchley2} and several others~\cite{churchley3}. 
A graph is \emph{$(k, l)$-partitionable} if it can be partitioned in up to $k$ cliques and $l$ independent sets with $k+l \geq 1$. Table \ref{table} contains trivial complexity results on $(k,l)$-partitionable problems in special classes of graphs for $k + l \leq 2$. In \cite{demange} efficient algorithms are devised for solving the $(k, l)$-partition problem on cographs, where $k$ and $l$ are finite. In \cite{klein} a characterization of $(k, l)$-partitionable cographs by forbidden induced subgraphs is provided, where $k$ and $l$ are finite. These results were later extended to $P_4$-sparse graphs \cite{bravo} and $P_4$-laden graphs \cite{bravo1}.
 
 \begin{table}[h]
\centering
    \begin{tabular}{|c|c|c|c|c|c|}
    \hline
    $k$ & $l$ & graph class & recognition & forbidden cographs & forbiden others\\ \hline
    $0$ & $1$ & edge-less & $\mathcal{O}(n)$ & $K_2$ & none \\ \hline
    $1$ & $0$ & complete & $\mathcal{O}(n+m)$ & $2K_1$ & none\\ \hline
    $1$ & $1$ & split & $\mathcal{O}(n+m)$ &$2K_2, C_4$ & $C_5$ \\ \hline
    $0$ & $2$ & bipartite & $\mathcal{O}(n+m)$ & $K_3$ & odd cycles\\ \hline
    $2$ & $0$ & co-bipartite & $\mathcal{O}(n+m)$ & $3K_1$ & odd co-cycles\\ \hline
    \end{tabular}
    \caption[Table]{Some trivial complexity results on $(k, l)$-partitionable problems}
    \label{table}
    \end{table}

Unless stated otherwise, we say that a graph $G = (V, E)$ is \emph{partitionable}, or has a \emph{partition}, if there exists a partition $\{A, B\}$ of $V$ such that $A$ induces a disjoint union of cliques, \ie, the cliques are vertex disjoint and have no edges between them, and $B$ induces a triangle-free graph. A graph is \emph{in-partitionable} if it is not partitionable. The class  of partitionable graphs generalises the classes of monopolar  and $(1, 2)$-partitionable graphs. In this paper we study the computational complexity of deciding whether a graph is partitionable. This problem is known to be $\NP$-complete on general graphs \cite{alastair}. We thus restrict our attention to special classes of graphs. Our hardness results are stated in the following theorem. 

\begin{theorem}\label{planarthm}
Let $G$ be a graph and let $\mathcal{C}$ be a finite set of cycles of lengths at least $5$. Then deciding whether $G$ is partitionable is $\NP$-complete whenever $G$ is bull-free, planar, perfect, $K_4$-free or $\mathcal{C}$-free.
\end{theorem} 

Theorem~\ref{planarthm} answers an open question posed by Thomass{\'e}, Trotignon and Vu\v{s}kovi{\'c}~\cite{vuskovic} about the complexity of our partition problem on bull-free graphs. 
We also show that the problem is tractable on the class of cographs. It is known that the relation of being an induced subgraph is a well-quasi-ordering on cographs \cite{damaschke}. Since the class of  partitionable cographs forms a subfamily of the class of cographs and is closed under induced subgraphs it follows that partitionable cographs  have a finite list of forbidden induced subgraphs. In this case it is folklore that deciding whether a cograph is partitionable can be done in polynomial-time. However, this proof of membership in $\PP$ is non-constructive. In our next theorem, we provide a constructive proof.

\begin{theorem}\label{mainthm}
A cograph $G$ is partitionable if and only if $G$ does not contain the graphs $H_1, H_2, \dots, H_{17}$ illustrated in Figure \ref{tabH}. 
\end{theorem}

We note that a result due to Stacho~\cite[Theorems 7.7 and 7.8 on pages 132-133]{stacho} shows that our partition problem is tractable on the class of chordal graphs. Moreover, the problem can be expressed in monadic second order logic without  edge set quantification. As a result it can be efficiently solved on graphs with bounded treewidth \cite{MSOLiii}, and bounded clique-width \cite{courcelle2}.

\section{Preliminaries}\label{background} 

All graphs considered here are finite and have no multiple edges and no loops. For undefined graph terminology we refer the reader to Diestel \cite{diestel}. 
Let $G = (V, E)$ be a graph and $V' \subseteq V$. The subgraph $G'$ induced by deleting the vertices $V \setminus V'$ from $G$ is denoted by $G' = G[V']$.  The graph $G \sm v$ is obtained from $G$ by deleting the vertex $v$. The complement of a graph $G$, denoted by $\ol{G}$, has the same vertex set as $G$ and two vertices in $\ol{G}$ are adjacent if and only if they are non-adjacent in $G$.  $K_n$, $C_n$, $P_n$ denote a complete graph, a cycle, and a path on $n$ vertices respectively. A \emph{universal vertex} is a vertex adjacent to every other vertex. An \emph{isolated vertex} is a vertex with no edges. 

A graph $G$ \emph{contains} a graph $H$ implies that $H$ is an induced subgraph of $G$. We say that $G$ is $H$-free if it contains no induced subgraph isomorphic to some graph $H$. Let $\mathcal{H}$ be a family of graphs. Then $G$ is $\mathcal{H}$-free if $G$ is $H$-free for each graph $H \in \mathcal{H}$. . We do not distinguish between isomorphic graphs.  The join $P=G \oplus H$ is formed from the disjoint graphs $G$ and $H$ by joining every vertex of $G$ to every vertex of $H$.   For three graphs $A, B$ and $C$, we have $A \oplus B \oplus C = (A \oplus B) \oplus C$.
The (disjoint) union $Q=G \cup H$ of disjoint graphs $G$ and $H$ has as vertex set $V(Q) = V(G) \cup V(H)$ and edge set $E(Q) = E(G) \cup E(H)$. If $G$ is a disconnected graph then it can be expressed as a union $G_1 \cup G_2 \cup \dots \cup G_k$, $k \geq 2$, of connected graphs. Moreover, each $G_i$ is a component of $G$ and each component is clearly a (vertex) maximal connected subgraph of $G$.  A $k$-colouring of a graph $G = (V, E)$ is a mapping $\phi: V \rightarrow \{1, \dots, k\}$ such that $\phi(u) \not= \phi(v)$ whenever $uv \in E$. A graph is bipartite if and only if it has a $2$-colouring.

 A graph is \textit{planar} if it can be drawn in the plane so that its edges intersect only at their ends. 
 An \textit{odd hole} is an induced cycle of odd length at least $5$. An \textit{odd antihole} is the complement of an odd hole. A graph $G$ is \textit{perfect} if for every induced subgraph $H$ of $G$, the chromatic number of $H$ equals the size of the largest clique of $H$. By the strong perfect graph theorem \cite{chudnovsky}, a graph is perfect if and only if it contains no odd hole and no odd antihole. A \emph{bull} is a self-complementary graph with degree sequence $(3,3,2,1,1)$. The class of cographs is equivalent to the class of $P_4$-free graphs \cite{cograph}.  It is well-known that a cograph or its complement is disconnected unless the cograph is $K_1$. A $P_3$-free graph is a union of cliques. A $\overline{P_3}$-free graph, or equivalently a $(K_2 \cup K_1)$-free graph, is a join of stable sets. Split graphs are exactly the $(1,1)$-partitionable graphs. They are characterized by the absence of
$2K_2$, $C_4$ and $C_5$. The intersection of
cographs and split graphs are the threshold graphs, characterised by
the absence of $2K_2$, $C_4$ and $P_4$. The diamond, paw, and butterfly graph can be written as $K_2 \oplus 2K_1$, $K_1 \oplus (K_1 \cup K_2)$ and $K_1 \oplus 2K_2$, respectively. The $k$-wheel graph is formed by a cycle $C$ of order $k - 1$ and a vertex not in $C$ with $k-1$ neighbours in $C$.  A $5$-wheel can be written as $C_4 \oplus K_1$, or $P_3 \oplus 2K_1$.

\section{Hardness results}

In this section we prove Theorem~\ref{planarthm}. Firstly we provide some gadgets that we will
use in reductions from 3SAT. Let $G=(V,E)$ be a graph and let $\{A,B\}$ be a partition of $V$ such
that $A$ induces a $K_3$-free subgraph of $G$ and $B$ induces a
$P_3$-free subgraph of $G$. For short we write that a vertex $v \in V$
is \emph{\Kf} if it belongs to $A$ and \emph{\Pf} if it belongs ot $B$.
A \emph{partition} is, unless stated otherwise, a partition of $V$
into \Kf and \Pf vertices.

\subsection{Negators}

A graph with two designated vertices $x$ and $y$ is a 
\emph{\Pf\ negator} if it has no partition where both $x$ and $y$ are
\Pf, but admits a partition where at most one of the vertices $x$ and $y$
is \Pf and the \Pf vertex has no \Pf neighbour. Examples of \Pf
negators are given in Figure \ref{fig:blue neg}. In what follows we implicitly use this partition.

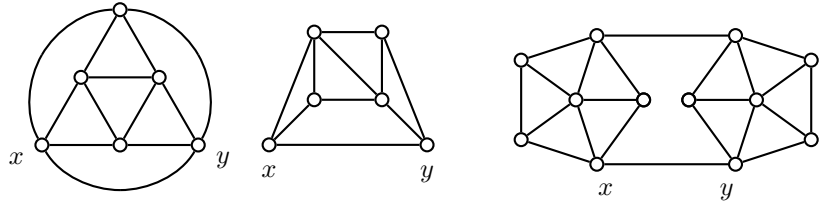
\begin{figure}[htbp]
  \psset{unit=6mm, radius=1mm}
  \hspace*{\fill}
  \begin{pspicture}(-2,-2)(+2,+2)
    \SpecialCoor \degrees[12] \psset{arcangle=-1}
    \multido{\na= 1+4, \nn=2+2}{3}{\Cnode(1;\na){i\nn}}
    \multido{\na= 1+4, \nn=2+2}{3}{\pnode(2;\na){o\nn}}
    \multido{\na=11+4, \nn=1+2}{3}{\Cnode(2;\na){o\nn}}
    \ncline{i2}{i4}    \ncline{i4}{i6}    \ncline{i6}{i2}
    \ncline{o1}{i2}    \ncline{i2}{o3}    \ncline{o3}{i4}
    \ncline{i4}{o5}    \ncline{o5}{i6}    \ncline{i6}{o1}
    \ncarc {o1}{o2}    \ncarc {o2}{o3}    \ncarc {o3}{o4}
    \ncarc {o4}{o5}    \ncarc {o5}{o6}    \ncarc {o6}{o1}
    \nput{7}{o5}{$x$}  \nput{11}{o1}{$y$}
  \end{pspicture}
  \hfill
  \begin{pspicture}(4,3.5)
    \Cnode(0.0,1.0){x} \nput{270}{x}{$x$}
    \Cnode(3.5,1.0){y} \nput{270}{y}{$y$} \ncline{x}{y}
    \Cnode(1.0,2.0){a} \ncline{a}{x}
    \Cnode(1.0,3.5){b} \ncline{b}{x} \ncline{b}{a}
    \Cnode(2.5,2.0){c} \ncline{c}{y} \ncline{c}{a} \ncline{c}{b}
    \Cnode(2.5,3.5){d} \ncline{d}{y} \ncline{d}{b} \ncline{d}{c}
  \end{pspicture}
  \hfill
  \begin{pspicture}(-3.5,-2)(+3.5,+2)
    \SpecialCoor \degrees[10]
    \rput(-2,0){\Cnode(0,0){l}
      \multido{\na=0+2, \nn=1+1}{6}{\Cnode(1.5;\na){l\nn}}
      \multido{\nm=1+1, \nn=2+1}{5}{\ncline{l\nm}{l\nn}\ncline{l\nm}{l}}
    }
    \rput(+2,0){\Cnode(0,0){r}
      \multido{\na=5+2, \nn=1+1}{6}{\Cnode(1.5;\na){r\nn}}
      \multido{\nm=1+1, \nn=2+1}{5}{\ncline{r\nm}{r\nn}\ncline{r\nm}{r}}
    }
    \ncline{l2}{r5}     \ncline{l5}{r2}
    \nput{8}{l5}{$x$}   \nput{7}{r2}{$y$}
  \end{pspicture}
  \hspace*{\fill}

  \caption{\Pf negators: the octahedron, the $P_6^2$-component 
    and the two-wheel}
  \label{fig:blue neg}
\end{figure}

Similarly, a graph with two designated vertices $x$ and $y$ is a 
\emph{\Kf\ negator} if it has no partition where both $x$ and $y$ are
\Kf, but admits a partition where at most one of the vertices $x$ and $y$
is \Kf and the \Pf vertex has no \Pf neighbour. We also implicitly use this partition.
 Examples of \Kf
negators are given in Figure \ref{fig:red neg}.

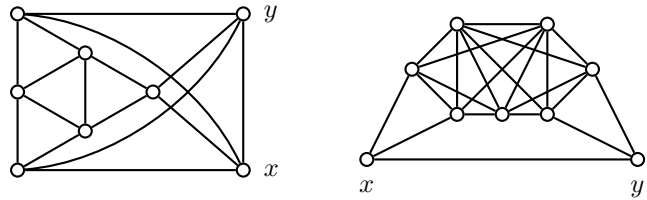
\begin{figure}[htbp]
  \psset{unit=6mm, radius=1mm}

  \hspace*{\fill}
  \begin{pspicture}(-2,-2)(+4,+2)
    \SpecialCoor \degrees[6] \psset{arcangle=-.5}
    \multido{\na=5+2, \nn=2+2}{3}{\Cnode(1;\na){i\nn}}
    \multido{\na=4+2, \nn=1+2}{3}{\Cnode(2;\na){o\nn}}
    \ncline{i2}{i4}    \ncline{i4}{i6}    \ncline{i6}{i2}
    \ncline{o1}{i2}    \ncline{i2}{o3}    \ncline{o3}{i4}
    \ncline{i4}{o5}    \ncline{o5}{i6}    \ncline{i6}{o1}
    \rput(3,0){\Cnode(2;5){x}} \nput{0}{x}{$x$}
    \ncline{x}{o1}     \ncline{x}{o3}     \ncarc {x}{o5}
    \rput(3,0){\Cnode(2;1){y}} \nput{0}{y}{$y$}
    \ncline{y}{o5}     \ncline{y}{o3}     \ncarc {o1}{y}
    \ncline{x}{y}
  \end{pspicture}
  \hfill
  \begin{pspicture}(0,-.5)(6,3.5)
    \Cnode(0,0){x}  \nput{270}{x}{$x$}
    \Cnode(6,0){y}  \nput{270}{y}{$y$} \ncline{x}{y}
    \Cnode(1,2){z1} \ncline{z1}{x}
    \Cnode(2,1){z2} \ncline{z2}{x} \ncline{z2}{z1}
    \Cnode(5,2){z3} \ncline{z3}{y}
    \Cnode(4,1){z4} \ncline{z4}{y} \ncline{z4}{z3}
    \Cnode(2,3){v1} 
    \Cnode(4,3){v2} \ncline{v1}{v2}
    \Cnode(3,1){v3} \ncline{v1}{v3} \ncline{v2}{v3}
    \multido{\nz=1+1}{4}{\multido{\nv=1+1}{3}{\ncline{z\nz}{v\nv}}}
  \end{pspicture}
  \hspace*{\fill}
  
  \caption{\Kf negators: the sun component and the bull-free component}
  \label{fig:red neg}
\end{figure}

Finally, a \emph{strong negator} is a graph that is both a \Kf negator
and a \Pf negator. Examples of strong negators, built from \Kf or \Pf
negators, are shown in Figure \ref{fig:strong neg}.

\begin{figure}[htbp]
  \psset{unit=9mm, radius=1mm}

  \hspace*{\fill}
  \begin{pspicture}(4,2)
    \SpecialCoor \degrees[6]
    \Cnode(2,1){d}
    \rput(1,1){\Cnode(0,0){c} \Cnode(1;2){a} \Cnode(1;4){b}}
    \rput(3,1){\Cnode(1;1){y} \Cnode(1;5){x}}
    \nput{0}{x}{$x$}    \nput{0}{y}{$y$}
    \ncline{a}{b} \ncline{b}{c} \ncline{c}{a}
    \ncline{d}{x} \ncline{d}{y}
    \psset{linestyle=dashed}
    \ncline{a}{d} \ncline{b}{d} \ncline{c}{d} \ncline{x}{y}
  \end{pspicture}
  \hfill
  \begin{pspicture}(2,2)
    \Cnode(0.25,1.75){b} \Cnode(1.75,1.75){y} \nput{0}{y}{$y$}
    \Cnode(0.25,0.25){a} \Cnode(1.75,0.25){x} \nput{0}{x}{$x$}
    \psset{linestyle=dashed}
    \ncline{a}{x} \ncline{x}{y} \ncline{y}{b} \ncline{b}{a}
  \end{pspicture}
  \hspace*{\fill}
  \caption{strong negators: The dashed lines represent \Pf negators in the 
    left graph and \Kf negators in the right. Their endpoints are the vertices
    $x$ and $y$ from these negators.}
  \label{fig:strong neg}
\end{figure}
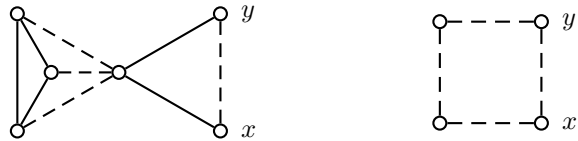

\subsection{Reduction from 3SAT}

We can now describe a generic reduction from 3SAT to our partition problem.
Let $\varphi$ be an instance of 3SAT, that is, a propositional formula in CNF
with clauses $c_1, c_2, \dots c_m$. Let $X = \{x_1, x_2,\dots,x_n\}$ be the
variables that occur in $\varphi$. We may safely assume that a variable and its negation do not occur in the same clause and that a variable does not occur more than once in the same clause. For every variable $x_i \in X$ we create a truth assignment component (tac) which is a ladder, whose edges are replaced by red or strong negators, with
$m$ rungs $x_{i,1}y_{i,1}, x_{i,2}y_{i,2}, \dots, \linebreak[1]
x_{i,m}y_{i,m}$, such that $\{x_{i,j} \mid 1 \le j \le m\}$ and
$\{y_{i,j} \mid 1 \le j \le m\}$ become independent sets in the tac. Note that the vertices $x$ and $y$ from the \Kf or strong negators that form the
ladder uniquely partition into two subsets, each of which can be either
\Kf or \Pf, see Figure \ref{fig:ladder}.
For every clause $c_j$ we create a satisfaction test component (stc) which is a $P_3$. For every
literal $x_i$ that appears in clause $c_j$ we identify the vertex
$x_{i,j}$ of the tac for $x_i$ with a vertex of the stc for $c_j$, and
the vertex $y_{i,j}$ of a tac is identified with a vertex from a stc
if $\neg x_i$ appears in $c_j$. This completes the construction of the
reduction graph $G$.

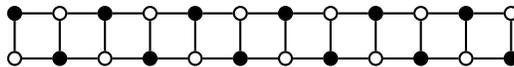
\begin{figure}[htbp]
  \centering
  \psset{unit=6mm, radius=1mm}
  \begin{pspicture}(11,1)
    \multido{\nx=0+2}{6}{\Cnode(\nx,0){l\nx} \Cnode*(\nx,1){u\nx}}
    \multido{\nx=1+2}{6}{\Cnode(\nx,1){u\nx} \Cnode*(\nx,0){l\nx}}
    \multido{\nx=0+1}{12}{\ncline{l\nx}{u\nx}}
    \multido{\nx=0+1, \nz=1+1}{11}{\ncline{l\nx}{l\nz} \ncline{u\nx}{u\nz}}
  \end{pspicture}
  \caption{A ladder with twelve rungs. In every partition,
    all black vertices belong to one part and all white vertices
    belong to the other part.
  }
  \label{fig:ladder}
\end{figure}

In case $\varphi$ is satisfiable we fix a satisfying truth assignment
of the variables in $X$. All true literals become \Kf and all false
literals become \Pf. Hence every tac is partitionable, and every stc contains at least one \Kf vertex and thus at most two (possibly adjacent) blue vertices with no other blue neighbours. This implies
 $G$ is partitionable.

Now let $G$ be partitionable. We assign the boolean value true to each 
variable $x_i$ with \Kf vertices representing the literal $x_i$ and
\Pf vertices representing $\neg x_i$, and false if the roles are the
other way around. This defines a consistent truth assignment for all
variables in $X$ because each tac is a ladder with at least two rungs.
We consider a clause $c_j$ of $\varphi$. It corresponds to a stc of $G$
which is a $P_3$. Hence one vertex of this stc is \Kf and therefore $c_j$
is satisfied.

\subsection{Planar graphs}

To show the \NP-completeness of the partition problem restricted to
planar graphs we reduce instead from planar 3SAT, which is \NP-complete~\cite{lichtenstein}, and use planar strong negators depicted in Figure \ref{fig:strong neg} whose dashed lines are either
blue negators from  Figure
\ref{fig:blue neg}. The fact that $G$ is planar can be easily derived from \cite{lichtenstein}: it suffices to contract every edge between a tac and a stc to obtain the associated (planar) graph of an instance of planar 3SAT.

\subsection{\boldmath$K_4$-free graphs}

The partition problem becomes trivial when restricted to triangle-free
graphs (these are graphs that do not contain $K_3$ as induced subgraph)
because all vertices can be made \Kf. Restricted to $K_4$-free graphs the
problem remains \NP-complete. 
The sun component in Figure \ref{fig:red neg} can be
used in the generic reduction from 3SAT we described above.

\subsection{Bull-free graphs}

The generic construction shows that the partition problem remains \NP-complete
when restricted to bull-free graphs: the graph $G$ is bull-free if the bull-free component from Figure
\ref{fig:red neg} is used in the generic reduction from 3SAT.

\subsection{Holes}

The other self-complementary graph on five vertices is $C_5$, a chordless
cycle. We will show that the partition problem remains \NP-complete for
$C_5$-free graphs, and more general for $\cC$-free graphs where $\cC$ is any
finite set of holes and a hole is a chordless cycle of length at least five.

Let $\cC$ be a finite set of holes and let $k$ be the length of the longest
cycle in $\cC$. We show that the problem remains \NP-complete for $\cC$-free graphs. In our generic reduction from 3SAT  we use
the sun component to build a ladder with $km$ rungs as tac for variable $x_i$.
As usual with \Kf negators, each stc is a $P_3$. If the literal $x_i$ appears
in clause $c_j$ we identify vertex $x_{i,jk}$ of the tac with a vertex of the
stc, and similarly for literal $\neg x_i$.

The reduction works as before. We only have to check that the reduction graph
is $\cC$-free. The sun component does not contain any holes. Hence both the
tac and the stc are hole-free. Therefore any hole in $G$ contains vertices
from different ladders. The distance between two vertices of a ladder that 
belong to other components is at least $k$. Hence $G$ does not contain
$C_5, C_6, \dots, C_k$.

\subsection{Perfect graphs}\label{perfectsection}

 Here we show $\NP$-completeness of the partition problem restricted to perfect graphs. Firstly we provide some gadgets that will be used in the reduction.

\subsubsection{Gadgets}

We use the $P_6^2$-component as the blue negator gadget shown in Figure~\ref{fig:blue neg},   and the strong negator gadget shown at the left of Figure~\ref{fig:strong neg}. 

The \textit{literal gadget} is shown in Figure \ref{fig:literal} where the double line symbolises the strong negator gadget. The gadget is partitionable and for every partition it has at least two blue endpoints. 

The \textit{propagator gadget}  is shown in Figure \ref{fig:literal}. The gadget is partitionable and for every partition it has exactly one or three blue endpoints.

Together the literal gadget and the propagator gadget form the satisfaction test component.

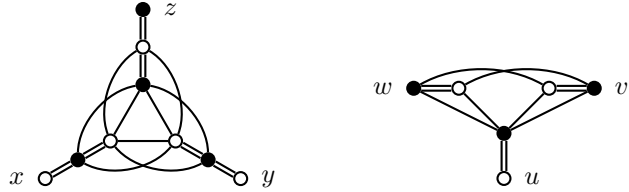
\begin{figure}[htbp]
  \centering
 \psset{radius=1mm, arcangle=40}
  \setlength{\tabcolsep}{5mm}
  \begin{tabular}{cc}

  \psset{unit=5mm}
    \begin{pspicture}(-7,-1.5)(0,4)
    \SpecialCoor \degrees[12]
    \Cnode (2;3){23}
      \multido{\nr=1+2}{2}{%
        \Cnode*(\nr;3){\nr3}
        \ncline[doubleline=true]{\nr3}{23}
      }
    \multido{\na=7+4}{2}{%
      \Cnode*(2;\na){2\na}
      \multido{\nr=1+2}{2}{%
        \Cnode (\nr;\na){\nr\na}
        \ncline[doubleline=true]{\nr\na}{2\na}
      }
    }
    \nput{6}{37}{$x$} \nput{0}{311}{$y$} \nput{0}{33}{$z$}
    \ncline{13}{17} \ncline{17}{111} \ncline{111}{13}
    \ncarc{17}{23}  \ncarc{23}{111}
    \ncarc{27}{13}  \ncarc{111}{27}
    \ncarc{211}{17} \ncarc{13}{211}
  \end{pspicture}
  &
  
   \psset{unit=2mm, radius=1mm}
    \begin{pspicture}(20,1)(-12,+5)

    \Cnode( 7,1){a}   \nput{0}{a}{$u$}
    \Cnode*( 7,4){a1}
   
    \Cnode(4, 7){c1}
    \Cnode*(1,7){c} \nput{180}{c}{$w$}
    
    \Cnode(10,7){b1}
    \Cnode*(13,7){b} \nput{360}{b}{$v$}
    
    \ncline{a1}{c} \ncline{a1}{c1} \ncline{a1}{b} \ncline{a1}{b1}
    
    \ncarc[arcangle=-30]{b1}{c} \ncarc[arcangle=30]{c1}{b}

    \psset{doubleline=true}
    \ncline{a}{a1} \ncline{b}{b1} \ncline{c}{c1}

  \end{pspicture}
  \end{tabular}
  \hspace*{\fill}  
  \caption{The literal gadget with endpoints $x, y, z$ and the propagator gadget with endpoints $u, v, w$ along with a partition where the white vertices are in the $P_3$-free part and the black vertices are in $K_3$-free part. Note that the propagator gadget is not symmetric.}
  \label{fig:literal}
\end{figure}

\subsubsection{Reduction from Positive 1-in-3-SAT} 

We describe a reduction from Positive $1$-in-$3$-SAT, which is known to be \NP-hard~\cite{schaefer},  to our partition problem on perfect graphs. An instance of Positive $1$-in-$3$-SAT is a set of  variables $X = \{x_1, x_2, \dots, x_n\}$ and a set of clauses $C = \{c_j \mid i=1,2,\dots,m\}$, such that each $c_j = (l_{j, 1} \lor l_{j, 2} \lor l_{j,3})$ consists of three positive literals and each literal $l_{j, k}$ is $x_i$ for some $x_i \in X$. The problem is to determine whether there exists a truth assignment to the variables in $X$ such that $\varphi = c_1 \land c_2 \land \dots \land c_m$ is satisfiable with exactly one true literal per clause.

For every variable $x_i \in X$ we create a truth assignment component (tac) which is a ladder, whose edges are strong negators, with
$m$ rungs $x_{i,1}y_{i,1}, x_{i,2}y_{i,2},\linebreak[1] \dots, 
x_{i,m}y_{i,m}$, such that the set $\{x_{i,j} \mid 1 \le j \le m\}$ of \emph{literal vertices} and
the set $\{y_{i,j} \mid 1 \le j \le m\}$ of \emph{propagator vertices}  become independent sets in the tac. Note that the vertices $x$ and $y$ from the strong negators that form the
ladder uniquely partition into two subsets, each of which can be either
\Kf or \Pf, see Figure \ref{fig:ladder}.
 For a clause $c = (x_1 \lor x_2 \lor x_3)$ where $x_1$, $x_2$ and $x_3$ are the
$i$'th, $j$'th and $k$'th occurrence, respectively, create a copy $H_c$ of the
literal gadget whose endpoints are identified with literal vertices $x_{1,i}, x_{2,j}$ and $x_{3,k}$,
and a copy $R_c$ of the propagator gadget whose endpoints are identified 
with propagator vertices $y_{1,i}, y_{1,j}$ and $y_{2,k}$. $H_c$ and $R_c$ are said to be the \textit{literal gadget} and \textit{propagator gadget}, respectively, of $C$. This completes the construction of the reduction graph $G$. 

\medskip

If $\varphi$ is satisfiable with exactly one true literal per clause we fix a satisfying truth assignment of the variables in $X$. All literal vertices corresponding to true literals become red and all literals vertices corresponding to false literals become blue. It is clear that every tac is partitionable. For a  literal gadget $H_c$ and a propagator gadget $R_c$ of clause $C$, $H_c$  has two blue endpoints  and $R_c$ has one blue endpoint, in which case $H_c$ and $R_c$ are partitionable. It follows that $G$ is partitionable.

Conversely, let $G$ be partitionable.  We assign boolean value true to each variable $x_i$ with red vertices representing the literal $x_i$,  and  false otherwise.  Consider the literal gadget $H_c$ and propagator gadget $R_c$ of clause $C$. If for contradiction all endpoints of $H_c$ are blue, then all endpoints of $R_c$ become red, a contradiction. Hence exactly one endpoint of $H_c$ is red, in which case $C$ has exactly one true literal, as required.

 \medskip 
We claim that $G$ is perfect. Let us first prove that $G$ contains no odd hole. The following observations follow by a careful examination of $G$.

\begin{itemize}
\item[(1)]The gadgets and tac are odd hole-free

\item[(2)]Each induced path between the endpoints of a literal or propagator gadget has even length
\end{itemize}

\noindent
 Let $C$ be an induced cycle of length at least $4$ in $G$. By (1), if $C$ is an induced subgraph of a gadget or tac then $C$ has even length. Otherwise, let  $R_1, \dots, R_k$ be induced subgraphs of tacs occurring on $C$ in that cyclic order. Clearly there exists a $2$-colouring $\phi$ of $R_1 \cup \dots \cup R_k$ where colour class $1$ are literal vertices and colour class $2$ are propagator vertices. It is easy to check that the segment $P_i$ of $C$ joining $R_i$ and $R_{i+1}$ is a path contained in a literal or propagator gadget whose endpoints are endpoints of that gadget. Since the endpoints of $P_i$ have the same colour under $\phi$ and  $P_i$ has even length by (2), $\phi$ can be extended to a $2$-colouring that includes $P_i$. This implies that $G$ contains no odd hole.

To prove that $G$ contains no odd antihole, we already established that $G$ does not contain $\ol{C_5} = C_5$.  Moreover, $G$ is $K_5$-free and hence $\ol{C_{2k+1}}$-free, $k \geq 5$. Now $K_4$ is contained in $\ol{C_7}$ (and hence $\ol{C_9}$). The only occurrences of $K_4$ in $G$ are in a literal or strong negator gadget. By considering adjacencies between such a $K_4$ and the rest of the graph it can be verified that $G$ does not contain $\ol{C_7}$ and $\ol{C_9}$.

\section{Cographs}\label{cographsection}

In this section we prove Theorem~\ref{mainthm}. We start by characterising  subclasses of partitionable cographs by forbidden induced subgraphs. These results will be useful in establishing the main theorem.

\subsection{Subclasses of partitionable cographs}

A  set of definitions and lemmas is initially required. A graph is \emph{bi-threshold} if it is bipartite or threshold. A graph is \emph{monopolar} if it is $(\infty, 1)$-partitionable. A graph is \emph{monopolar nearly split} if it is monopolar or $(1, 2)$-partitionable.




\begin{lemma}\label{first}
Let $G$ be a graph. If $G$ contains $P_3$ and $K_3$, then $G$ contains $F_1 = P_3 \cup K_3$, $F_2 = \mathrm{diamond}$, or $F_3 = \mathrm{paw}$. 
\end{lemma}

\begin{proof}
Consider the triangle. If there is a vertex with exactly one or two neighbours in the triangle we have $F_3$ or $F_2$, respectively. If two non-adjacent vertices with three neighbours in the triangle exist we have $F_2$. If none of these cases applies to any triangle in $G$ then all triangles form a clique with no neighbours in the rest of the graph. Consequently we find $F_1$.
 \end{proof}

\begin{lemma}\label{second}
Let $G$ be a cograph. If $G$ contains $P_3$ and $2K_2$, then $G$ contains $Q_1 = P_3 \cup K_2$, or $Q_2 = \mathrm{butterfly}$. 
\end{lemma}

\begin{proof}
Consider the disjoint edges $e_1$ and $e_2$ in $2K_2$. Let $G_1$ be the component that contains $e_1$. First suppose $G_1$ contains $e_2$. Let $v$ be a vertex adjacent to some endpoint of $e_1$ and on a path between $e_1$ and $e_2$. Since $G$ is a cograph any induced path between two vertices in a component of $G$ has length at most $2$. As $e_1$ and $e_2$ have no edges between them every induced path between $e_1$ and $e_2$ has length $2$. It follows that $v$ must be adjacent to every vertex in $e_1$ and $e_2$, in which case $Q_2$ is found. 
Finally suppose $G_1$ does not contain $e_2$. If there is a vertex with exactly one neighbour in $e_1$ then $Q_1$ is obtained. If this case does not apply to any vertex in $G_1$ then $G_1$ forms a clique with no neighbours in the rest of the graph and $Q_1$ is again obtained.
 \end{proof}

\begin{lemma}\label{third}
Let $G$ be a cograph. If $G$ is $C_4$-free and contains $P_3, 2K_2$ and $K_3$, then $G$ contains $S_1 = F_1$, $S_2 = Q_2$, $S_3 = K_2 \cup \mathrm{paw}$, or $S_4 = K_2 \cup \mathrm{diamond}$. 
\end{lemma}

\begin{proof}
Consider the disjoint edges $e_1$ and $e_2$ in $2K_2$. Let $G_1$ be the component containing $e_1$. If $G_1$ contains $e_2$ then, by the same argument as in the proof of Lemma \ref{second}, we find $S_2$. So suppose $G_1$ does not contain $e_2$. We distinguish a number of cases. 
If there exists two non-adjacent vertices with two neighbours in $e_1$ then $S_4$ is obtained. 
If there exists two non-adjacent vertices with one and two neighbours, respectively, in $e_1$ then $S_3$ is obtained. 
If there exists two adjacent vertices with one and two neighbours, respectively, in $e_1$ then $S_4$ is found.  If none of these cases applies to any edge in $G_1$ then, by considering the absence of $P_4$ and $C_4$, $G_1$ either forms (i) a  star graph with no neighbours in the rest of the graph, or (ii) a clique with no neighbours in the rest of the graph. In the case of (i) we find $S_1$. In the case of (ii) if $G_1$ contains a triangle then $S_1$ is obtained and if $G_1$ is a single edge we find $S_1, S_3$ or $S_4$, by Lemma \ref{first}. 
 \end{proof}

\begin{lemma}\label{fourth}
Let $G$ be a cograph. If $G$ contains $P_3$ and $2K_3$, then $G$ contains $W_1 = 2K_3 \cup P_3$, $W_2 = K_3 \cup \mathrm{diamond}$, $W_3 = K_3 \cup \mathrm{paw}$, or $W_4 = K_1 \oplus 2K_3$. 
\end{lemma}

\begin{proof}
Consider the disjoint triangles $t_1$ and $t_2$ in $2K_3$. If $t_1$ and $t_2$ share a neighbour then, by considering the absence of $P_4$, $W_4$ is obtained. Otherwise, by a similar argument to that in Lemma \ref{first}, we find $W_1$, $W_2$, or $W_3$.  
 \end{proof}

\subsubsection{Bi-threshold cographs}\label{1}

This section establishes the following theorem.

\begin{theorem}\label{btheorem}
Let $G$ be a  connected cograph. Then $G$ is bi-threshold if and only if $G$ does not contain the graphs $B_1, \dots, B_6$ depicted in Figure \ref{btgraphs}.
\end{theorem}

(1) $B_1 = \mathrm{butterfly}$.

(2) $B_2 = C_4 \oplus K_1$.

(3) $B_3 = 2K_1 \oplus (K_2 \cup K_1)$.

(4) $B_4 = K_2 \cup \mathrm{diamond}$.

(5) $B_5 =  K_3 \cup P_3$.

(6) $B_6 = K_2 \cup \mathrm{paw}$.

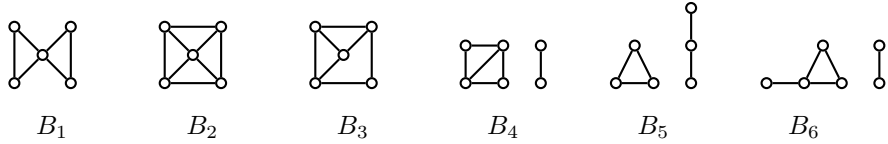
\begin{figure}[htbp]
  \centering
  \psset{radius=0.8mm, arcangle=30}
  \setlength{\tabcolsep}{5mm}
  \begin{tabular}{ccccccccc}

       \psset{unit=5mm}
       \begin{pspicture}(0,-.5)(2,2.5)
       \Cnode(0, 0){0}
       \Cnode(0, 1.5){1}
       \Cnode(0.75, 0.75){c}
       \Cnode(1.5, 0){2}
       \Cnode(1.5, 1.5){3}

     \ncline{0}{c} \ncline{1}{c} \ncline{2}{c} \ncline{3}{c} \ncline{2}{3} \ncline{0}{1}

      \end{pspicture}
       &
       \psset{unit=5mm}
       \begin{pspicture}(0,-.5)(2,2.5)
      \multido{\nx=0+1.5}{2}{\multido{\ny=0+1.5}{2}{\Cnode(\nx,\ny){\nx\ny}}}

      \Cnode(0.75, 0.75){c}

       \multido{\nx=0+1.5}{2}{\multido{\ny=0+1.5}{2}{ \ncline{\nx\ny}{c}}}

       \ncline{00}{1.50}
       \ncline{00}{01.5}
       \ncline{1.51.5}{01.5} \ncline{1.51.5}{1.50}

       \end{pspicture}
       &
       \psset{unit=5mm}
    \begin{pspicture}(0,-.5)(2,2.5)
      \multido{\nx=0+1.5}{2}{\multido{\ny=0+1.5}{2}{\Cnode(\nx,\ny){\nx\ny}}}

      \Cnode(0.75, 0.75){c}

       
       \ncline{00}{c} \ncline{01.5}{c} \ncline{1.51.5}{c}

       \ncline{00}{1.50}
       \ncline{00}{01.5}
       \ncline{1.51.5}{01.5} \ncline{1.51.5}{1.50}
         
       \end{pspicture}
       %
       &
       \psset{unit=5mm}
    \begin{pspicture}(0,-.5)(2,2.5)
      \Cnode(0, 0){00}
      \Cnode(0,1){01}
      \Cnode(1,0){10}
      \Cnode(1,1){11}
      
      \Cnode(2,0){20}
      \Cnode(2,1){21}


       \ncline{00}{11}
       \ncline{00}{10}
       \ncline{00}{01}
       \ncline{11}{01} \ncline{11}{10}
       
       \ncline{20}{21}
         
       \end{pspicture}
       &
       \psset{unit=5mm}
    \begin{pspicture}(0,-.5)(2,2.5)
      \Cnode(0, 0){00}
      \Cnode(1,0){10}
      \Cnode(0.5,1){x}
      \Cnode(2,0){u}
      \Cnode(2, 1){v}
      \Cnode(2, 2){w}

      \ncline{00}{10} \ncline{00}{x} \ncline{10}{x} \ncline{u}{v} \ncline{v}{w}

       \end{pspicture}
      
      & 
      \psset{unit=5mm}
    \begin{pspicture}(0,-.5)(2,2.5)
      \Cnode(0, 0){00}
      \Cnode(1,0){10}
      \Cnode(2,0){20}
      \Cnode(1.5,1){x}
      
      \Cnode(3,0){30}
      \Cnode(3,1){31}
      
      \ncline{30}{31}
      
      \ncline{00}{10} \ncline{10}{20} \ncline{x}{10} \ncline{x}{20}

       \end{pspicture}

       \\
       $B_1$ & $B_2$ & $B_3$ & $B_4$ & $B_5$ & $B_6$

       \end{tabular}
  \caption{The graphs $B_1, B_2, B_3, B_4, B_5$ and $B_6$}
  \label{btgraphs}
\end{figure}

\begin{proof}
($\Leftarrow$) Recall that a threshold graph is $(C_4, P_4, 2K_2)$-free and a bipartite graph is triangle-free. But the graphs $B_1, \dots, B_6$ each contain a triangle, and $C_4$ or $2K_2$. 

($\Rightarrow$) Let $G$ be a connected cograph that is neither bipartite nor threshold and vertex minimal. If $G$ is complete the result is easily seen to be true. So suppose that $G$ contains $P_3$. In particular $G$ must contain $K_3$, and $C_4$ or $2K_2$. We distinguish two cases.

\medskip
\noindent
\textbf{Case 1}:  $G$ contains $C_4$.
\medskip

\noindent
Since $G$ is connected and $P_4$-free there exists a triangle and a quadrangle that share
an edge. The third vertex of the triangle has another neighbour in the
quadrangle, otherwise there would be a $P_4$. Consequently $G$ contains $B_2$ or $B_3$.

\medskip
\noindent
\textbf{Case 2}: $G$ contains $2K_2$.
\medskip

\noindent
By Lemma \ref{third}, $G$ contains $B_1, B_4$, $B_5$ or $B_6$. This completes the proof.
 \end{proof}

\subsubsection{Monopolar cographs}\label{2}

In \cite{ekim} a forbidden induced subgraph characterization of monopolar cographs, defined in the paper as $(s, k)$-polar cographs where $\min(s, k) \leq 1$, is presented. (Note that our definition of monopolar graphs is different). Essentially, the same proof shows the following result.

\begin{theorem}\label{monopolarthm1}
Let $G$ be a connected cograph. Then $G$ is monopolar if and only if $G$ has no induced subgraph isomorphic to the graphs $J_1$, \dots, $J_4$ depicted in Figure \ref{jgraphs}. \end{theorem}

(1) $J_1 = 5-\mathrm{wheel}$.


(2) $J_2 = K_1 \oplus (P_3 \cup K_2)$.

(3) $J_3 = K_2 \oplus 2K_2$.

(4) $J_4 = (K_2 \cup K_1) \oplus (K_2 \cup K_1)$.

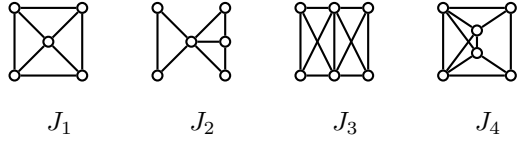
\begin{figure}[htbp]
  \centering
  \psset{radius=0.8mm, arcangle=30}
  \setlength{\tabcolsep}{3.5mm}
  \begin{tabular}{ccccccc}
\psset{unit=6mm}
    \begin{pspicture}(0,-.5)(2,2.5)
      \multido{\nx=0+1.5}{2}{\multido{\ny=0+1.5}{2}{\Cnode(\nx,\ny){\nx\ny}}}

      \Cnode(0.75, 0.75){c}

       \multido{\nx=0+1.5}{2}{\multido{\ny=0+1.5}{2}{ \ncline{\nx\ny}{c}}}

       \ncline{00}{1.50}
       \ncline{00}{01.5}
       \ncline{1.51.5}{01.5} \ncline{1.51.5}{1.50}

       \end{pspicture}
       &
       \psset{unit=6mm}
       \begin{pspicture}(0,-.5)(2,2.5)
       \Cnode(0, 0){00}
       \Cnode(0, 1.50){01.50}
       \Cnode(0.75, 0.75){c}

      \multido{\ny=0+0.75}{3}{\Cnode(1.5,\ny){1.50\ny}}

       \multido{\nx=0+0.75}{3}{\multido{\ny=0+0.75}{3}{\ncline{\nx\ny}{c}}}

       \ncline{00}{01.50} \ncline{1.500}{1.500.75} \ncline{1.500.75}{1.501.50}

       \end{pspicture}
       &
       \psset{unit=6mm}
    \begin{pspicture}(0,-.5)(2,2.5)
      \multido{\nx=0+0.75}{3}{\multido{\ny=0+1.5}{2}{\Cnode(\nx,\ny){\nx\ny}}}

       \ncline{00}{01.5} \ncline{0.750}{0.751.5} \ncline{1.500}{1.501.5}
       \multido{\nx=0+1.50}{2}{\multido{\ny=0+1.5}{2}{\ncline{0.750}{\nx\ny} \ncline{0.751.5}{\nx\ny}}}
             
       \end{pspicture}
       &
       \psset{unit=6mm}
    \begin{pspicture}(0,-.5)(2,2.5)
      \multido{\nx=0+1.5}{2}{\multido{\ny=0+1.5}{2}{\Cnode(\nx,\ny){\nx\ny}}}

      \Cnode(0.75, 0.50){c1}
      \Cnode(0.75, 1){c2}

       \multido{\nx=0+1.5}{2}{\ncline{0\nx}{c1} \ncline{0\nx}{c2}}
       \ncline{1.50}{c1} \ncline{1.51.5}{c2}

       \ncline{c1}{c2}
       \ncline{00}{1.50}
       \ncline{00}{01.5}
       \ncline{1.51.5}{01.5} \ncline{1.51.5}{1.50}
         
       \end{pspicture}

        \\
       $J_1$ & $J_2$ & $J_3$ & $J_4$

       \end{tabular}
  \caption{The graphs $J_1, J_2, J_3$ and $J_4$}
  \label{jgraphs}
\end{figure}

\begin{proof}
($\Leftarrow$) Recall that a monopolar graph is a graph that can be partitioned into an independent set and a union of cliques. Since every $J_i$ is not a union of cliques, it must contain a join of stable sets in any partition. It is routine to verify that there exists no partition of these graphs such that their join of stable sets in the partition is a stable set. 

($\Rightarrow$) Since $G$ is connected it is the join of two cographs $G[A]$ and $G[B]$. Since a threshold graph is $(C_4, P_4, 2K_2)$-free,  it suffices to consider the following cases.

\medskip
\noindent
\textbf{Case 1}: $G[A]$ is not a threshold graph.
\medskip

\noindent
\textbf{Subcase 1.1}: $G[A]$ contains $C_4$.
\medskip

\noindent
Since $G[B]$ is non-empty, $G$ contains $J_1$.   

\medskip
\noindent
\textbf{Subcase 1.2}: $G[A]$ contains $2K_2$.
\medskip  

\noindent
If $G[B]$ contains $K_2$ then $G$ contains $J_3$. So suppose $G[B]$ is a stable set. 
If $G[A]$ contains $P_3$ then,  by Lemma \ref{second},  $G[A]$ contains $Q_1$ or $Q_2$. If $G[A]$ contains $Q_2$ then $G$ contains $J_3 = Q_2 \oplus K_1$, and if $G[A]$ contains $Q_1$ then $G$ contains $J_2 = Q_1 \oplus K_1$. 
Finally if $G[A]$ is $P_3$-free then $G = G[A] \oplus G[B]$ is  monopolar. This completes Case 1.

It may be assumed by symmetry that both $G[A]$ and $G[B]$ do not contain $C_4$, $2K_2$ and $P_4$ and hence form threshold graphs.

\medskip
\noindent
\textbf{Case 2}: $G[A]$ and $G[B]$ are threshold graphs.
\medskip

\noindent
\textbf{Subcase 2.1}: $G[A]$ contains a triangle.
\medskip

\noindent
\textit{(1)} If $G[A]$ is a clique then $G[B]$ being a threshold graph, $G$ is also a threshold graph and  therefore monopolar. 

\noindent
\textit{(2)} Suppose $G[A]$ contains a paw or a diamond. In both cases $G[A]$ contains $P_3$. If $G[B]$ contains $2K_1$ then $G$ contains $J_1 = P_3 \oplus 2K_1$, and if $G[B]$ is a clique then $G$ is a threshold graph. 

\noindent
\textit{(3)} Suppose $G[A]$ contains at least one isolated vertex besides the triangle. If $G[B]$ contains $P_3$ then $G$ contains
$J_1 = P_3 \oplus 2K_1$. So we may assume that $G[B]$ is a union of cliques. If $G[B]$ contains $K_2 \cup K_1$ then $G$ contains $J_4 = (K_2 \cup K_1) \oplus (K_2 \cup K_1)$. If $G[B]$ is a non-trivial stable set then $G$ is monopolar. Finally if $G[B]$ is a clique then $G$ forms a threshold graph.

\medskip
\noindent
\textbf{Subcase 2.2}: Both $G[A]$ and $G[B]$ are triangle-free.
\medskip

\noindent
\textit{(1)} Suppose $G[A]$ contains $P_3$. If $G[B]$ contains $2K_1$ then $G$ contains
$J_1 = P_3 \oplus 2K_1$. If $G[B]$ is a clique then $G$ is a threshold graph.

\noindent
\textit{(2)} We may thus assume, by symmetry, that $G[A]$ and $G[B]$ are $P_3$-free. 

First suppose $G[A]$ contains $K_2 \cup K_1$. If $G[B]$ contains $K_2 \cup K_1$ then $G$ contains $J_4 = (K_2 \cup K_1) \oplus (K_2 \cup K_1)$. So let $G[B]$ be $(K_2 \cup K_1)$-free. If $G[B]$  is a stable set then $G$ is monopolar. Otherwise, $G[B]$ is a clique in which case $G$ is a threshold graph. Second suppose $G[A]$ is a clique. Since $G[B]$ is a threshold graph, it follows that $G$ is a threshold graph. 
Finally if $G[A]$ is a stable set, $G[B]$ being $P_3$-free it follows that $G$ is monopolar. This completes the proof.
 \end{proof}

\begin{remark}\label{monopolarcoro}
The graphs $J_1, J_2, J_3$ and $J_4$ are $(1, 2)$-partitionable connected cographs. 
\end{remark}

\begin{proof} If $C(J_i)$ denotes a maximum clique of $J_i$, $i = 1, \dots, 4$, then $J_i[V \setminus C(J_i)]$ is bipartite.  
 \end{proof}

\subsubsection{Monopolar nearly split cographs}\label{3}

In this section we prove Theorem~\ref{almostsplitthm}. First we need an auxiliary result. 

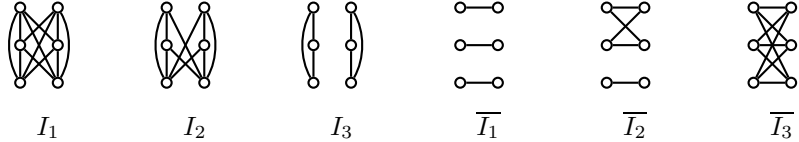
\begin{figure}[htbp]
  \centering
  \psset{radius=0.8mm, arcangle=30}
  \setlength{\tabcolsep}{3.5mm}
  \begin{tabular}{cccccc}
  \psset{unit=5mm}
    \begin{pspicture}(-0.5,-.5)(2,2.5)
      \multido{\nx=0+1}{2}{\multido{\ny=0+1}{3}{\Cnode(\nx,\ny){\nx\ny}}}
      
      \ncline{00}{01}
      \ncarc{00}{02}
      \ncline{00}{11}
      \ncline{00}{12}
      
      \ncline{01}{02}
      \ncline{01}{10}
      \ncline{01}{12}
      
      \ncline{02}{10}
      \ncline{02}{11}
      
      \ncline{10}{11}
      \ncarc[arcangle=-30]{10}{12}
      \ncline{11}{12}
      
      \end{pspicture}
      &
       \psset{unit=5mm}
    \begin{pspicture}(-0.5,-.5)(2,2.5)
      \multido{\nx=0+1}{2}{\multido{\ny=0+1}{3}{\Cnode(\nx,\ny){\nx\ny}}}
      
      \ncline{00}{01}
      \ncarc{00}{02}
      \ncline{00}{11}
      \ncline{00}{12}
      
      \ncline{01}{02}
      \ncline{01}{10}

      \ncline{02}{10}

      \ncline{10}{11}
      \ncarc[arcangle=-30]{10}{12}
      \ncline{11}{12}
      
      \end{pspicture}
      &
       \psset{unit=5mm}
    \begin{pspicture}(-0.5,-.5)(2,2.5)
      \multido{\nx=0+1}{2}{\multido{\ny=0+1}{3}{\Cnode(\nx,\ny){\nx\ny}}}
      
      \ncline{00}{01}
      \ncarc{00}{02}
       \ncline{01}{02}
       \ncline{10}{11}
      \ncarc[arcangle=-30]{10}{12}
      \ncline{11}{12}

      \end{pspicture}

  %
   &
    \psset{unit=5mm}
    \begin{pspicture}(-0.5,-.5)(2,2.5)
      \multido{\nx=0+1}{2}{\multido{\ny=0+1}{3}{\Cnode(\nx,\ny){\nx\ny}}}
      \ncline{00}{10}
      \ncline{01}{11}
      \ncline{02}{12}
      
    \end{pspicture}
    &
    \psset{unit=5mm}
    
    \begin{pspicture}(-0.50,-.5)(2,2.5)
      \multido{\nx=0+1}{2}{\multido{\ny=0+1}{3}{\Cnode(\nx,\ny){\nx\ny}}}
      \ncline{00}{10}
      \ncline{01}{11}
      \ncline{02}{12}
      \ncline{01}{12}
      \ncline{02}{11}
    \end{pspicture}
    &
    \psset{unit=5mm}
    \begin{pspicture}(-0.5,-.5)(2,2.5)
      \multido{\nx=0+1}{2}{\multido{\ny=0+1}{3}{\Cnode(\nx,\ny){\nx\ny}}}
      \ncline{00}{10}
      \ncline{01}{11}
      \ncline{02}{12}
      \ncline{01}{12}
      \ncline{02}{11}
      \ncline{00}{11}
      \ncline{00}{12}
      \ncline{10}{01}
      \ncline{10}{02}
    \end{pspicture}
    \\
    $I_1$ & $I_2$ & $I_3$ & $\ol{I_1}$ & $\ol{I_2}$ & $\ol{I_3}$
     \end{tabular}
  \caption{The graphs $I_1, I_2, I_3$ and their complements}
  \label{fgraphs}
\end{figure}

\begin{proposition}[\cite{demange}]\label{proposition1}
A cograph is $(2, 1)$-partitionable if and only if it does not contain the graphs $\ol{I_1}$, $\ol{I_2}$, $\ol{I_3}$ depicted in Figure \ref{fgraphs}.
\end{proposition}

\begin{corollary}\label{corigraphs}
A cograph is $(1, 2)$-partitionable if and only if it does not contain the graphs $I_1 = \ol{3K_2}, I_2 = 2K_2 \oplus 2K_1 , I_3 = 2K_3$ depicted in Figure \ref{fgraphs}.
\end{corollary}

We are now ready to prove the theorem.

\begin{theorem}\label{almostsplitthm}

Let $G$ be a connected cograph. Then $G$ is a monopolar nearly split graph if and only if $G$ does not contain the graphs $R_1, \dots, R_8$ depicted in Figure \ref{rgraphs}.
\end{theorem}

(1) $R_1 = 2K_1 \oplus 2K_1 \oplus 2K_1$.

(2) $R_2 = 2K_2 \oplus (K_2 \cup K_1)$.

(3) $R_3 = 2K_1 \oplus (P_3 \cup K_2)$.

(4) $R_4 = K_1 \oplus (2K_1 \oplus 2K_2)$. 

(5) $R_5 = K_2 \oplus 2K_3$.

(5') $R_5 = K_1 \oplus (K_1 \oplus 2K_3)$.

(6) $R_6 = K_1 \oplus (P_3 \cup 2K_3)$.

(7) $R_7 = K_1 \oplus (K_3 \cup (P_3 \oplus K_1))$.

(8) $R_8 = K_1 \oplus (K_3 \cup (K_1 \oplus (K_1 \cup K_2)))$.

\hbox{}

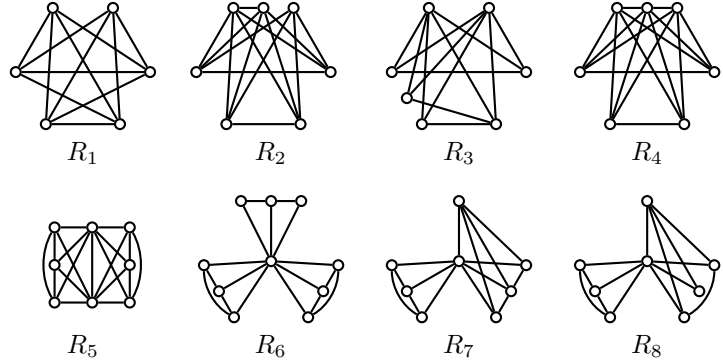
\begin{figure}[htbp]
  \centering
  \psset{radius=0.8mm, arcangle=30}
  \setlength{\tabcolsep}{4.5mm}
  \begin{tabular}{cccccc}
    \psset{unit=4mm}
    \begin{pspicture}(-2,-2)(+2,+2)
      \SpecialCoor
      \rput(2; 90){\multido{\ix=-1+2,\nx=0+2}{2}{\Cnode(\ix;  0){0\nx}}}
      \rput(2;210){\Cnode(-1;120){10} \Cnode(1;120){12}}
      \rput(2;330){\Cnode(-1;240){20} \Cnode(1;240){22}}
      \multido{\ny=0+2}{3}{
        \ncline{0\ny}{10} \ncline{0\ny}{12}
        \ncline{0\ny}{20} \ncline{0\ny}{22}
      }
      \ncline{10}{20} \ncline{10}{22} \ncline{12}{20} \ncline{12}{22}
    \end{pspicture}
      &
      \psset{unit=4mm}
    \begin{pspicture}(-1.5,-2)(+2,+2)
      \SpecialCoor
      \rput(2; 90){\multido{\ix=-1+1,\nx=0+1}{3}{\Cnode(\ix;  0){0\nx}}}
       \rput(2;210){\Cnode(-1; 120){10} \Cnode(1; 120){12}}
      \rput(2;330){\Cnode(-1;240){20} \Cnode(1;240){22}}
     
      \multido{\nx=0+1}{3}{
        \ncline{0\nx}{10} \ncline{0\nx}{12}
        \ncline{20}{0\nx} \ncline{22}{0\nx}
      }
      \ncline{10}{22} \ncline{12}{20} 
      
      \ncline{00}{01}
    \end{pspicture}
      &
      \psset{unit=4mm}
    \begin{pspicture}(-2,-2)(+2,+2)
      \SpecialCoor
      \rput(2; 90){\multido{\ix=-1+2,\nx=0+2}{2}{\Cnode(\ix;  0){0\nx}}}
      \rput(2;210){\multido{\ix=-1+1,\nx=0+1}{3}{\Cnode(\ix;120){1\nx}}}
      \rput(2;330){\Cnode(-1;240){20} \Cnode(1;240){22}}
      \multido{\nx=0+2}{2}{
        \multido{\ny=0+1}{3}{\ncline{0\nx}{1\ny}}
        \ncline{20}{0\nx} \ncline{22}{0\nx}
      }
      \ncline{10}{22} \ncline{12}{20}  \ncline{11}{22}
    \end{pspicture}
    &
     \psset{unit=4mm}
    \begin{pspicture}(-2,-2)(+2,+2)
      \SpecialCoor
      \rput(2; 90){\multido{\ix=-1+1,\nx=0+1}{3}{\Cnode(\ix;  0){0\nx}}}
      \rput(2;210){\multido{\ix=-1+2,\nx=0+2}{2}{\Cnode(\ix;120){1\nx}}}
      \rput(2;330){\Cnode(-1;240){20} \Cnode(1;240){22}}
      \multido{\nx=0+1}{1}{\ncline{\nx0}{\nx1}\ncline{\nx1}{\nx2}}
      \multido{\nx=0+1}{3}{
        \multido{\ny=0+2}{2}{\ncline{0\nx}{1\ny}}
        \ncline{20}{0\nx} \ncline{22}{0\nx}
      }
      \ncline{10}{22} \ncline{12}{20} 
    \end{pspicture}

  \\
  $R_1$ & $R_2$ & $R_3$ & $R_4$
   \\
   \\

    
    \psset{unit=5mm}
    \begin{pspicture}(-0.5,-.5)(2,2.5)
      \multido{\nx=0+2}{2}{\multido{\ny=0+1}{3}{\Cnode(\nx,\ny){\nx\ny}}}
      
      \Cnode(1,0){10}
      \Cnode(1,2){12}
      
      \ncline{10}{12}
      
      \ncline{00}{01}
      \ncarc{00}{02}
       \ncline{01}{02}
       \ncline{20}{21}
      \ncarc[arcangle=-30]{20}{22}
      \ncline{21}{22}
      
      \multido{\nx=0+2}{2}{\multido{\ny=0+1}{3}{\ncline{\nx\ny}{10} \ncline{\nx\ny}{12}}}

      \end{pspicture}

    &
     \psset{unit=4mm}
    \begin{pspicture}(-2,-2)(+2,+2)
      \SpecialCoor
      \rput(2; 90){\multido{\ix=-1+1,\nx=0+1}{3}{\Cnode(\ix;  0){0\nx}}}
      \rput(2;210){\multido{\ix=-1+1,\nx=0+1}{3}{\Cnode(\ix;120){1\nx}}}
      \rput(2;330){\multido{\ix=-1+1,\nx=0+1}{3}{\Cnode(\ix;240){2\nx}}}
      
      \Cnode(0,0){x}
      
      \multido{\nx=0+1}{3}{\multido{\ny=0+1}{3}{\ncline{\nx\ny}{x}}}
      
      \ncline{00}{01} \ncline{01}{02}
      \ncline{10}{11} \ncline{11}{12} \ncarc{10}{12}
      \ncline{20}{21} \ncline{21}{22} \ncarc{20}{22}
      
    \end{pspicture}
    &
    \psset{unit=4mm}
    \begin{pspicture}(-2,-2)(+2,+2)
      \SpecialCoor
      \rput(2; 90){\Cnode(0;  0){01}}
      \rput(2;210){\multido{\ix=-1+1,\nx=0+1}{3}{\Cnode(\ix;120){1\nx}}}
      \rput(2;330){\multido{\ix=-1+1,\nx=0+1}{3}{\Cnode(\ix;240){2\nx}}}
      
      \Cnode(0,0){x}
      
      \multido{\nx=0+1}{3}{\ncline{2\nx}{01} \ncline{2\nx}{x} \ncline{1\nx}{x} }
      \ncline{01}{x}

      
      \ncline{10}{11} \ncline{11}{12} \ncarc{10}{12}
      \ncline{20}{21} \ncline{21}{22} 
      
    \end{pspicture}
    &
    \psset{unit=4mm}
    \begin{pspicture}(-2,-2)(+2,+2)
      \SpecialCoor
      \rput(2; 90){\Cnode(0;  0){01}}
      \rput(2;210){\multido{\ix=-1+1,\nx=0+1}{3}{\Cnode(\ix;120){1\nx}}}
      \rput(2;330){\multido{\ix=-1+1,\nx=0+1}{3}{\Cnode(\ix;240){2\nx}}}
      
      \Cnode(0,0){x}
      
      \multido{\nx=0+1}{3}{\ncline{2\nx}{01} \ncline{2\nx}{x} \ncline{1\nx}{x} }
      \ncline{01}{x}

      
      \ncline{10}{11} \ncline{11}{12} \ncarc{10}{12}
      \ncarc{20}{22}  
      
    \end{pspicture}
    
    \\
    $R_5$ & $R_6$ & $R_7$ & $R_8$
   
     \end{tabular}
  \caption{The graphs $R_1, \dots, R_8$}
  \label{rgraphs}
\end{figure}

\begin{proof}
($\Leftarrow$) This is proved by a careful case analysis. 

($\Rightarrow$) Suppose $G$ is neither monopolar nor $(1,2)$-partitionable and vertex minimal. Since $G$ is connected it is the join of two cographs $G[A]$ and $G[B]$. By the minimality of $G$, $G[A]$ and $G[B]$ are either monopolar or $(1,2)$-partitionable. We distinguish a number of cases.

\medskip
\noindent
\textbf{Case 1}: $G[A]$ and $G[B]$ are ($K_2 \cup K_1$)-free. 
\medskip

\noindent
It follows that $G$ is a join of stable sets. Hence $G$ either contains $R_1 = \ol{3K_2}$ or is $(1, 2)$-partitionable.

\medskip
\noindent
\textbf{Case 2}: $G[A]$ and $G[B]$ contain $K_2 \cup K_1$. 
\medskip

\noindent
\textit{(1)} If $G[A]$ contains $C_4$ then $G$ contains $R_1 = C_4 \oplus 2K_1$. 

\noindent
\textit{(2)} If $G[A]$ contains $2K_2$ then $G$ contains $R_2 = 2K_2 \oplus (K_2 \cup K_1)$. 

\noindent
\textit{(3)} By symmetry if $G[A]$ and $G[B]$ are threshold graphs then $G$ is $(1, 2)$-partitionable.

\medskip
\noindent
\textbf{Case 3: } $G[A]$ is $(K_2 \cup K_1)$-free, and $G[B]$ contains $K_2 \cup K_1$.
\medskip

\noindent
\textbf{Subcase 3.1}: $G[A]$ is a clique.
\medskip

\noindent
 If $G[B]$ is $(1, 2)$-partitionable then $G$ is $(1,2)$-partitionable. Otherwise, $G[B]$ must be monopolar. By Corollary \ref{corigraphs} and given that $J_1 \subset I_1$ it follows that $G[B]$ contains $I_2$ or $I_3$.
 
 \noindent
\textit{(1)} If $G[B]$ contains $I_2$ then $G$ contains $R_4 = K_1 \oplus I_2$. 
 
 \noindent
\textit{(2)} Suppose $G[B]$ contains $I_3$. If $G[A]$ has at least 2 vertices then $G$ contains $R_5 = K_2 \oplus I_3$. So suppose $G[A]$ is a single vertex. If $G[B]$ is $P_3$-free then $G$ is monopolar. If $G[B]$ contains $P_3$ then, by Lemma \ref{fourth}, $G[B]$ contains $W_1, W_2, W_3$ or $W_4$. It follows that $G$ contains $R_6 = K_1 \oplus W_1$, $R_7 = K_1 \oplus W_2$, $R_8 = K_1 \oplus W_3$, or $R_5 = K_1 \oplus W_4$, respectively.

\medskip
\noindent
\textbf{Subcase 3.2}: $G[A]$ is an independent set.
\medskip

\noindent
The case where $G[A]$ is a single vertex is covered in Subcase 3.1. We may thus assume that $G[A]$ contains $2K_1$. If $G[B]$ is $P_3$-free  then $G$ is monopolar. If $G[B]$ is a threshold graph  then $G$ is $(1, 2)$-partitionable. Otherwise, $G[B]$ contains $C_4$, or $P_3$ and $2K_2$. If $G[B]$ contains $C_4$ then $G$ contains $R_1 = 2K_1 \oplus C_4$. If $G[B]$ contains $P_3$ and $2K_2$ then, by Lemma \ref{second}, $G[B]$ contains $Q_1$ or $Q_2$. Hence $G$ contains $R_3 = 2K_1 \oplus Q_1$ or $R_4 = 2K_1 \oplus Q_2$, respectively.

\medskip
\noindent
\textbf{Subcase 3.3}: $G[A]$ contains $2K_1 \oplus 2K_1$. 
\medskip

\noindent
Since $G[B]$ contains $K_2 \cup K_1$, it follows that $G$ contains $R_1 = 2K_1 \oplus 2K_1 \oplus 2K_1$. 

\medskip
\noindent
\textbf{Subcase 3.4}: $G[A] = qK_1 \oplus K_r$ for some integers $q \geq 2$ and $r \geq 1$.
\medskip 

\noindent
If $G[B]$ is a threshold graph then G is $(1, 2)$-partitionable. Otherwise, $G[B]$ contains $2K_2$ or $C_4$. It follows that $G$ either contains $R_4$ or $R_1$, respectively. This completes the proof.
 \end{proof}

\subsection{Main result}\label{4}

This section establishes Theorem \ref{mainthm}. The following two lemmas are first required. The first lemma is implicit in~\cite{ekim}.

\begin{lemma}\label{connectedobs}
Minimal in-partitionable cographs are connected. 
\end{lemma}


\begin{proof}
Let $G=(V,E)$ be a cograph. Suppose to the contrary that $G$ is disconnected and vertex minimal in-partitionable. Let $\{A,B\}$ be a partition of $V$ such that $G = G[A] \cup G[B]$. By the minimality of $G$, $G[A]$ and $G[B]$ are partitionable. Let $C$ and $D$ be a partition of $G[A]$, $P$ and $Q$ a partition of $G[B]$ such that $G[C]$, $G[P]$ are bipartite, and $G[D]$, $G[Q]$ are $P_3$-free. It follows that $G[C\cup P]$ is bipartite and $G[D \cup Q]$ is $P_3$-free, which is a partition of $G$.
 \end{proof}

\begin{lemma}\label{basiclemma1}
Let $G=(V,E)$ be a cograph, and let $\{A,B\}$ be a partition of $V$ such that $G = G[A] \oplus G[B]$. If both $G[A]$ and $G[B]$ are threshold graphs then $G$ is partitionable.  
\end{lemma}

\begin{proof}
Let $G' = G[A]$ and $G'' = G[B]$. Let $\{C, D\}$ be a partition of $V(G')$ such that $C$ induces a clique and $D$ induces a stable set. Similarly, let $\{F, P\}$ be a partition of $V(G'')$ such that $F$ induces a clique and $G$ induces a stable set. Since $G = G[A] \oplus G[B]$, it follows that $G[C \cup F] = G[C] \oplus G[F]$ is a clique and $G[D \cup P ] = G[D] \oplus G[P]$ is a complete bipartite graph. 
 \end{proof}

\begin{figure}[htbp]
  \centering
 \psset{radius=0.8mm, arcangle=30}
  \setlength{\tabcolsep}{4mm}
  \begin{tabular}{cccccc}
    \psset{unit=4mm}
    \begin{pspicture}(-2,-2)(+2,+2)
      \SpecialCoor
      \rput(2; 90){\multido{\ix=-1+1,\nx=0+1}{3}{\Cnode(\ix;  0){0\nx}}}
      \rput(2;210){\Cnode(-1;120){10} \Cnode(1;120){12}}
      \rput(2;330){\Cnode(-1;240){20} \Cnode(1;240){22}}
      \ncline{00}{01} \ncline{01}{02}
      \multido{\ny=0+1}{3}{
        \ncline{0\ny}{10} \ncline{0\ny}{12}
        \ncline{0\ny}{20} \ncline{0\ny}{22}
      }
      \ncline{10}{20} \ncline{10}{22} \ncline{12}{20} \ncline{12}{22}
    \end{pspicture}
    &
    \psset{unit=4mm}
    \begin{pspicture}(-2,-2)(+2,+2)
      \SpecialCoor
      \rput(2; 90){\multido{\ix=-1+1,\nx=0+1}{3}{\Cnode(\ix;  0){0\nx}}}
      \rput(2;210){\multido{\ix=-1+1,\nx=0+1}{3}{\Cnode(\ix;120){1\nx}}}
      \rput(2;330){\Cnode(-1;240){20} \Cnode(1;240){22}}
      \multido{\nx=0+1}{2}{\ncline{\nx0}{\nx1}\ncline{\nx1}{\nx2}}
      \multido{\nx=0+1}{3}{
        \multido{\ny=0+1}{3}{\ncline{0\nx}{1\ny}}
        \ncline{20}{0\nx} \ncline{22}{0\nx}
      }
      \ncline{10}{22} \ncline{12}{20} \ncline{11}{20} \ncline{11}{22}
    \end{pspicture}
    &
    
    \psset{unit=4mm}
    \begin{pspicture}(-2,-2)(+2,+2)
      \SpecialCoor
      \rput(2; 90){\multido{\ix=-1+1,\nx=0+1}{3}{\Cnode(\ix;  0){0\nx}}}
      \rput(2;210){\multido{\ix=-1+1,\nx=0+1}{3}{\Cnode(\ix;120){1\nx}}}
      \rput(2;330){\Cnode(-1;240){20} \Cnode(1;240){22}}
      \multido{\nx=0+1}{2}{\ncline{\nx0}{\nx1}}
      \multido{\nx=0+1}{3}{
        \multido{\ny=0+1}{3}{\ncline{0\nx}{1\ny}}
        \ncline{20}{0\nx} \ncline{22}{0\nx}
      }
      \ncline{10}{22} \ncline{10}{20} \ncline{12}{20}
      \ncline{12}{22} \ncline{11}{20} \ncline{11}{22}
    \end{pspicture}
     &
     
    \psset{unit=4mm}
    \begin{pspicture}(-2,-2)(+2,+2)
      \SpecialCoor
      \rput(2; 90){\multido{\ix=-1+1,\nx=0+1}{3}{\Cnode(\ix;  0){0\nx}}}
      \rput(2;210){\multido{\ix=-1+1,\nx=0+1}{3}{\Cnode(\ix;120){1\nx}}}
      \rput(2;330){\Cnode(-1;240){20} \Cnode(1;240){22}}
      \multido{\nx=0+1}{1}{\ncline{\nx0}{\nx1}\ncline{\nx1}{\nx2}}
      \multido{\nx=0+1}{3}{
        \multido{\ny=0+1}{3}{\ncline{0\nx}{1\ny}}
        \ncline{20}{0\nx} \ncline{22}{0\nx}
      }
      \ncline{10}{22} \ncline{12}{20}  \ncline{11}{22}
    \end{pspicture}

    \\
    $H_1$ & $H_2$ & $H_3$ & $H_4$
    \\
    \\

    \psset{unit=4mm}
    \begin{pspicture}(-2,-2)(+2,+2)
      \SpecialCoor
      \rput(2; 90){\multido{\ix=-1+1,\nx=0+1}{3}{\Cnode(\ix;  0){0\nx}}}
      \rput(2;210){\multido{\ix=-1+1,\nx=0+1}{3}{\Cnode(\ix;120){1\nx}}}
      \rput(2;330){\Cnode(-1;240){20} \Cnode(1;240){22}}
      \multido{\nx=0+1}{3}{
        \multido{\ny=0+1}{3}{\ncline{0\nx}{1\ny}}
        \ncline{20}{0\nx} \ncline{22}{0\nx}
      }
      \ncline{10}{22} \ncline{12}{20} \ncline{11}{20} \ncline{11}{22}
      \ncline{10}{11} \ncline{11}{12}
      \ncline{00}{01}
    \end{pspicture}
    
    &
    \psset{unit=4mm}
    \begin{pspicture}(-2,-2)(+2,+2)
      \SpecialCoor
      \rput(2; 90){\multido{\ix=-1+1,\nx=0+1}{3}{\Cnode(\ix;  0){0\nx}}}
      \rput(2;210){\multido{\ix=-1+1,\nx=0+1}{3}{\Cnode(\ix;120){1\nx}}}
      \rput(2;330){\multido{\ix=-1+1,\nx=0+1}{3}{\Cnode(\ix;240){2\nx}}}
      \multido{\nx=0+1}{3}{\multido{\ny=0+1}{3}{\ncline{0\nx}{1\ny} \ncline{0\nx}{2\ny}}}
      
      \ncline{10}{11} \ncline{11}{12} \ncarc{10}{12}
      \ncline{20}{21} \ncline{21}{22}
      
      \ncarc{00}{02}
    \end{pspicture}
    
    &
    \psset{unit=4mm}
    \begin{pspicture}(-2,-2)(+2,+2)
      \SpecialCoor
      \rput(2; 90){\multido{\ix=-1+1,\nx=0+1}{3}{\Cnode(\ix;  0){0\nx}}}
      \rput(2;210){\multido{\ix=-1+1,\nx=0+1}{3}{\Cnode(\ix;120){1\nx}}}
      \rput(2;330){\multido{\ix=-1+1,\nx=0+1}{3}{\Cnode(\ix;240){2\nx}}}
      \multido{\nx=0+1}{3}{\multido{\ny=0+1}{3}{\ncline{0\nx}{1\ny} \ncline{0\nx}{2\ny}}}
      \ncarc{00}{02}
      \ncarc{20}{22}
      \ncline{10}{11} \ncline{11}{12} 
      \multido{\nx=0+1}{3}{\ncline{1\nx}{21}}
    \end{pspicture}
    
    &
    \psset{unit=4mm}
    \begin{pspicture}(-2,-2)(+2,+2)
      \SpecialCoor
      \rput(2; 90){\multido{\ix=-1+1,\nx=0+1}{3}{\Cnode(\ix;  0){0\nx}}}
      \rput(2;210){\multido{\ix=-1+1,\nx=0+1}{3}{\Cnode(\ix;120){1\nx}}}
      \rput(2;330){\multido{\ix=-1+1,\nx=0+1}{3}{\Cnode(\ix;240){2\nx}}}
       \multido{\nx=0+1}{3}{\multido{\ny=0+1}{3}{\ncline{0\nx}{1\ny} \ncline{0\nx}{2\ny}}}
       \ncarc{00}{02}
       \ncline{10}{11} \ncline{11}{12} \ncarc{10}{12}
       \ncarc{20}{22}
       \ncline{21}{11}
    \end{pspicture}

        \\
    $H_5$ & $H_6$ & $H_7$ & $H_8$
    \\
    \\
   \psset{unit=6mm}
    \begin{pspicture}(0,-.5)(2,2.5)
      \multido{\nx=0+1.5}{2}{\multido{\ny=0+1.5}{2}{\Cnode(\nx,\ny){\nx\ny}}}
      \multido{\ny=0+0.75}{3}{\Cnode(3,\ny){3\ny}}
      
      \Cnode(0.75, 0.75){c}
      \Cnode(2.25, 0.75){v}
      
       \multido{\nx=0+1.5}{2}{\multido{\ny=0+1.5}{2}{\ncline{\nx\ny}{v} \ncline{\nx\ny}{c}}}
       \multido{\ny=0+0.75}{3}{\ncline{3\ny}{v}}
       \ncline{v}{c}
       \ncline{00}{1.50}
       \ncline{00}{01.5}
       \ncline{1.51.5}{01.5} \ncline{1.51.5}{1.50}

       \ncline{30}{30.75} \ncarc[arcangle=-30]{30}{31.50} \ncline{30.75}{31.50}
       
       \end{pspicture}
       &
       \psset{unit=6mm}
       \begin{pspicture}(-1,-.5)(2,2.5)
       \Cnode(0, 0){00}
       \Cnode(0, 1.50){01.50}
       \Cnode(0.75, 0.75){c}

      \multido{\ny=0+0.75}{3}{\Cnode(1.5,\ny){1.50\ny}}
      \multido{\ny=0+0.75}{3}{\Cnode(3,\ny){3\ny}}

      \Cnode(2.25, 0.75){v}
      
       \multido{\nx=0+0.75}{3}{\multido{\ny=0+0.75}{3}{\ncline{\nx\ny}{v} \ncline{\nx\ny}{c}}}
       \multido{\ny=0+0.75}{3}{\ncline{3\ny}{v}}

       \ncline{00}{01.50} \ncline{1.500}{1.500.75} \ncline{1.500.75}{1.501.50}
       
       \ncline{30}{30.75} \ncarc[arcangle=-30]{30}{31.50} \ncline{30.75}{31.50}
       
       \end{pspicture}
       &
       \psset{unit=6mm}
    \begin{pspicture}(-1,-.5)(2,2.5)
      \multido{\nx=0+0.75}{3}{\multido{\ny=0+1.5}{2}{\Cnode(\nx,\ny){\nx\ny}}}
      \multido{\ny=0+0.75}{3}{\Cnode(3,\ny){3\ny}}
      
      \Cnode(2.25, 0.75){v}
      
       \multido{\nx=0+0.75}{3}{\multido{\ny=0+1.5}{2}{\ncline{\nx\ny}{v}}}
       \multido{\ny=0+0.75}{3}{\ncline{3\ny}{v}}
       
       
       \ncline{00}{01.5} \ncline{0.750}{0.751.5} \ncline{1.500}{1.501.5}
       \multido{\nx=0+1.50}{2}{\multido{\ny=0+1.5}{2}{\ncline{0.750}{\nx\ny} \ncline{0.751.5}{\nx\ny}}}

       \ncline{30}{30.75} \ncarc[arcangle=-30]{30}{31.50} \ncline{30.75}{31.50}
       
       \end{pspicture}
       &
         \psset{unit=6mm}
    \begin{pspicture}(-1,-.5)(2,2.5)
      \multido{\nx=0+1.5}{2}{\multido{\ny=0+1.5}{2}{\Cnode(\nx,\ny){\nx\ny}}}
      \multido{\ny=0+0.75}{3}{\Cnode(3,\ny){3\ny}}
      
      \Cnode(0.75, 0.50){c1}
      \Cnode(0.75, 1){c2}
      \Cnode(2.25, 0.75){v}
      
       \multido{\nx=0+1.5}{2}{\multido{\ny=0+1.5}{2}{\ncline{\nx\ny}{v}}}
       \multido{\nx=0+1.5}{2}{\ncline{0\nx}{c1} \ncline{0\nx}{c2}}
       \ncline{1.50}{c1} \ncline{1.51.5}{c2}
       \multido{\ny=0+0.75}{3}{\ncline{3\ny}{v}}
       \ncline{v}{c1}
       \ncline{v}{c2}
       \ncline{c1}{c2}
       \ncline{00}{1.50}
       \ncline{00}{01.5}
       \ncline{1.51.5}{01.5} \ncline{1.51.5}{1.50}

       \ncline{30}{30.75} \ncarc[arcangle=-30]{30}{31.50} \ncline{30.75}{31.50}
       
       \end{pspicture}

       \\
       $H_{9}$ & $H_{10}$ & $H_{11}$ & $H_{12}$
    \\
    \\
    \psset{unit=4mm}
    \begin{pspicture}(-2,-2)(+2,+2)
      \SpecialCoor
      \rput(2; 90){\multido{\ix=-1+1,\nx=0+1}{3}{\Cnode(\ix;  0){0\nx}}}
      \rput(2;210){\multido{\ix=-1+1,\nx=0+1}{3}{\Cnode(\ix;120){1\nx}}}
      \rput(2;330){\multido{\ix=-1+1,\nx=0+1}{3}{\Cnode(\ix;240){2\nx}}}
      
     \Cnode(0,-0.5){x}
      \Cnode(0, 0.5){y}
      \ncline{x}{y}
      
      \multido{\nx=1+1}{2}{\multido{\ny=0+1}{3}{\ncline{\nx\ny}{x} \ncline{\nx\ny}{y}}}
      
      \ncline{00}{x} \ncline{02}{x} \ncline{00}{y} \ncline{01}{y} \ncline{02}{y} \ncarc[arcangle=-30]{01}{x}
      
      \ncline{00}{01} \ncline{01}{02}
      \ncline{10}{11} \ncline{11}{12} \ncarc{10}{12}
      \ncline{20}{21} \ncline{21}{22} \ncarc{20}{22}
      
    \end{pspicture}
    &
    \psset{unit=4mm}
    \begin{pspicture}(-2,-2)(+2,+2)
      \SpecialCoor
      \rput(2; 90){\Cnode(0;  0){01}}
      \rput(2;210){\multido{\ix=-1+1,\nx=0+1}{3}{\Cnode(\ix;120){1\nx}}}
      \rput(2;330){\multido{\ix=-1+1,\nx=0+1}{3}{\Cnode(\ix;240){2\nx}}}
      
      \Cnode(0,-0.5){x}
      \Cnode(0, 0.5){y}
      
      \multido{\nx=0+1}{3}{\ncline{2\nx}{01} \ncline{2\nx}{x} \ncline{1\nx}{x} \ncline{2\nx}{y} \ncline{1\nx}{y} }
      \ncarc[arcangle=-30]{01}{x}
      \ncline{01}{y}
      \ncline{x}{y}

      
      \ncline{10}{11} \ncline{11}{12} \ncarc{10}{12}
      \ncline{20}{21} \ncline{21}{22} 
      
    \end{pspicture}
    &
    \psset{unit=4mm}
    \begin{pspicture}(-2,-2)(+2,+2)
      \SpecialCoor
      \rput(2; 90){\Cnode(0;  0){01}}
      \rput(2;210){\multido{\ix=-1+1,\nx=0+1}{3}{\Cnode(\ix;120){1\nx}}}
      \rput(2;330){\multido{\ix=-1+1,\nx=0+1}{3}{\Cnode(\ix;240){2\nx}}}
      
      \Cnode(0,-0.5){x}
      \Cnode(0, 0.5){y}
      
      \multido{\nx=0+1}{3}{\ncline{2\nx}{01} \ncline{2\nx}{x} \ncline{1\nx}{x} \ncline{2\nx}{y} \ncline{1\nx}{y} }
      \ncarc[arcangle=-30]{01}{x}
      \ncline{01}{y}
      
      \ncline{x}{y}

      
      \ncline{10}{11} \ncline{11}{12} \ncarc{10}{12}
      \ncarc{20}{22}  
      
    \end{pspicture}
    &
    
    \psset{unit=5mm}
    \begin{pspicture}(-0.5,-.5)(2,2.5)
      \multido{\nx=0+1}{3}{\multido{\ny=0+1}{3}{\Cnode(\nx,\ny){\nx\ny}}}
      \multido{\nx=0+1}{3}{
        \ncline{\nx0}{\nx1} \ncline{\nx1}{\nx2} \ncarc{\nx0}{\nx2}
      }
   \multido{\nx=0+1,\nz=1+1}{2}{\multido{\ny=0+1}{3}{\multido{\nu=0+1}{2}{\ncline{\nx\ny}{\nz\nu}}}}
      \ncline{00}{12} \ncline{01}{12} \ncline{02}{12}
      \ncarc{00}{22}  \ncarc {01}{22} \ncarc {02}{22}
    \end{pspicture}
    &
    
    \psset{unit=5mm}
    \begin{pspicture}(-0.5,-.5)(2,2.5)
      \multido{\nx=0+1}{3}{\multido{\ny=0+1}{3}{\Cnode(\nx,\ny){\nx\ny}}}
      \multido{\nx=0+1}{3}{
        \ncline{\nx0}{\nx1}\ncline{\nx1}{\nx2}\ncarc{\nx0}{\nx2}
      }

  \multido{\nx=0+1,\nz=1+1}{2}{\multido{\ny=0+1}{3}{\multido{\nu=0+1}{3}{\ncline{\nx\ny}{\nz\nu}}}}

    \end{pspicture}
    \\
    $H_{13}$ & $H_{14}$ & $H_{15}$ & $H_{16}$ & $H_{17}$

  \end{tabular}
  \caption{Forbidden subgraphs of partitionable cographs.}
  \label{tabH}
\end{figure}
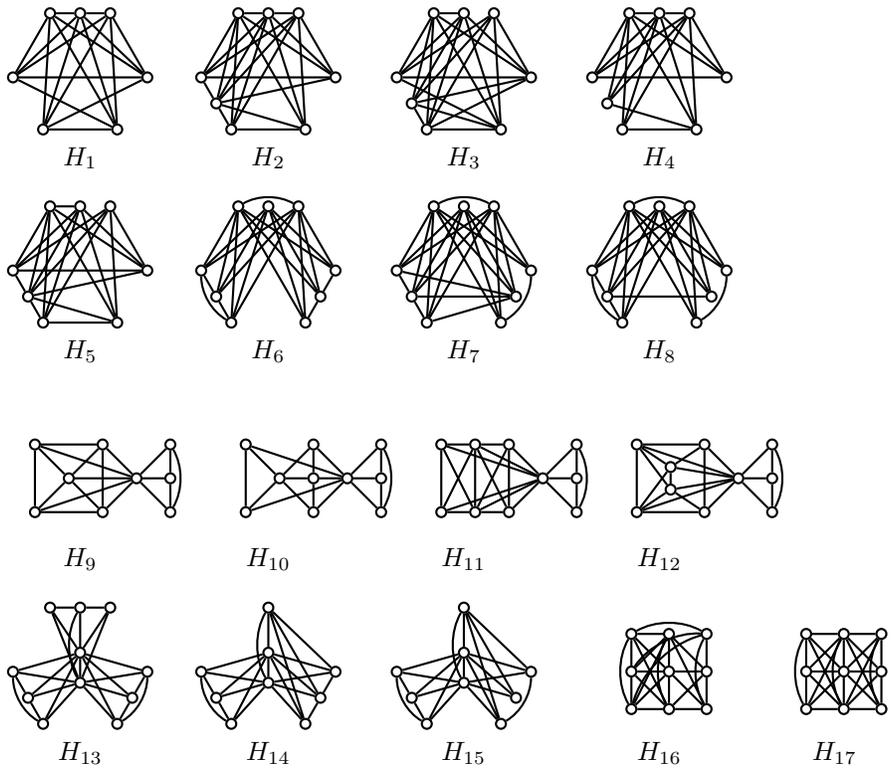

The following graphs depicted in Figure \ref{tabH} will be used:
\begin{itemize}
\itemsep0em
    \item[(1)] $H_1 = 2K_1 \oplus 2K_1 \oplus 2K_1 \oplus K_1$
    \item [(2)] $H_2 = P_3 \oplus K_1 \oplus 2K_2$
    \item [(3)] $H_3= 2K_1 \oplus(K_2 \cup K_1)\oplus(K_2 \cup K_1)$
    \item [(4)] $H_4 = P_3 \oplus (K_2 \cup P_3)$
    \item [(5)] $H_5 = (K_2 \cup K_1) \oplus K_1 \oplus 2K_2$
    \item [(6)] $H_6 = (K_2 \cup K_1) \oplus (K_3 \cup P_3)$
    \item [(7)] $H_7 = (K_2 \cup K_1) \oplus (K_2 \cup (P_3 \oplus K_1))$
    \item [(8)] $H_8 = (K_2 \cup K_1) \oplus (K_2 \cup (K_1 \oplus (K_2 \cup K_1)))$
    \item [(9)] $H_9 = K_1 \oplus (K_3 \cup (C_4 \oplus K_1))$ 
    \item [(10)] $H_{10} = K_1 \oplus (K_3 \cup (K_1 \oplus (P_3 \cup K_2)))$
    \item [(11)] $H_{11} = K_1 \oplus (K_3 \cup (K_2 \oplus 2K_2))$
    \item [(12)] $H_{12} = K_1 \oplus (K_3 \cup ((K_2 \cup K_1) \oplus (K_2 \cup K_1)))$ 
    \item [(13)] $H_{13} = K_2 \oplus (P_3 \cup 2K_3)$
    \item [(14)] $H_{14} = K_2 \oplus (K_3 \cup (P_3 \oplus K_1))$
    \item [(15)] $H_{15} = K_2 \oplus (K_3 \cup (K_1 \oplus (K_1 \cup K_2))$
    \item [(16)] $H_{16} = (K_3 \cup K_2) \oplus (K_3 \cup K_1)$
    \item [(17)] $H_{17} = K_3 \oplus 2K_3$
\end{itemize}

\medskip

We make a simple observation. It is folklore that a graph is bipartite if and only if it contains no odd cycle. Since a cograph contains no odd hole, it follows that  a cograph $G = (V, E)$ is partitionable if and only if there exists a partition $\{A, B\}$ of $V$ such that $A$ induces a $P_3$-free graph and $B$ induces a bipartite graph.

We are now ready to prove the theorem.

\begin{proof}[Proof of Theorem \ref{mainthm}] ($\Leftarrow$) This follows by a careful case analysis.

($\Rightarrow$) Suppose $G$ is vertex minimal in-partitionable. By Lemma \ref{connectedobs}, $G$ is connected. We prove that $G$ must contain one of the graphs $H_1, \dots, H_{17}$. 

\begin{claim}\label{c1}
If $G$ has no universal vertex then $G$ contains one of the graphs $H_1, \dots, H_8$, $H_{16}$.
\end{claim}

\begin{proof}[Proof of Claim~\ref{c1}]
Since $G$ is connected it is the join of two cographs $G[A]$ and $G[B]$. By the minimality of $G$, $G[A]$ and $G[B]$ are partitionable. Since $G$ has no universal vertex, $G[A]$ and $G[B]$ have no universal vertex. Consequently $G[A]$ and $G[B]$ each contain $2K_1$. We consider a number of  cases.

\medskip
\noindent
\textbf{Case 1}: $G[A]$ is $P_3$-free.
\medskip

\noindent
$G[A]$ is a union of at least two cliques $C_1, C_2$ because it contains $2K_1$. 

\medskip
\noindent
\textbf{Subcase 1.1}: $G[B]$ is $P_3$-free. 
\medskip

\noindent
Similarly $G[B]$ is a union of at least two cliques $C_3, C_4$. If $G[B]$ or $G[A]$ is bipartite then $G$ is partitionable. So it may be assumed, without loss of generality, that $|C_1|, |C_3| \geq 3$. Moreover $C_2$ or $C_4$ contains $K_2$, otherwise $G[A]$ and $G[B]$ form threshold graphs and $G$ is partitionable by Lemma \ref{basiclemma1}. We imply that $G$ contains $H_{16} = (K_3 \cup K_2) \oplus (K_3 \cup K_1)$.

\medskip
\noindent
\textbf{Subcase 1.2}: $G[B]$ contains $P_3$. 
\medskip

\noindent
\textit{(1)} $G[A]$ is a stable set of order at least two. 

\noindent
If $G[B]$ is monopolar then $G$ is partitionable. Otherwise, by Theorem \ref{monopolarthm1}, $G[B]$ contains one of the graphs $J_1, J_2, J_3, J_4$. It follows that $G$ contains $H_1 = 2K_1 \oplus J_1$, $H_2= 2K_1 \oplus J_3$, $H_3 = 2K_1 \oplus J_4$, or $H_4 = 2K_1 \oplus J_2$, respectively. 

\noindent
\textit{(2)} $G[A]=K_r \cup K_1$ for some integer $r \geq 2$. 

\noindent
 If $G[B]$ is a threshold graph then $G$ is $(1, 2)$-partitionable. If $G[B]$ is bipartite then $G$ is partitionable. Otherwise, \ie $G[B]$ contains $K_3$, and $C_4$ or $2K_2$, by Theorem \ref{btheorem}, $G[B]$ contains one of the graphs $B_1, B_2, B_3, B_4, B_5$ or $B_6$.  It follows that $G$ contains $H_{5} = (K_2 \cup K_1) \oplus B_1$, $H_1 = 2K_1 \oplus B_2$, $H_3 = (K_2 \cup K_1) \oplus B_3$, $H_{7} = (K_2 \cup K_1) \oplus B_4$, $H_{6} = (K_2 \cup K_1) \oplus B_5$, or $H_{8} = (K_2 \cup K_1) \oplus B_6$, respectively.

\noindent
\textit{(3)} $G[A]$ contains $2K_2$.

\noindent
If $G[B]$ is bipartite then $G$ is partitionable. We may thus assume that $G[B]$ contains a triangle (as $G[B]$ is a cograph). Since $G[B]$ contains $P_3$, by Lemma \ref{first}, $G[B]$ contains one of the graphs $F_1, F_2, F_3$. It follows that $G$ contains $H_{6} = (K_2 \cup K_1) \oplus F_1$, $H_2 = 2K_2 \oplus F_2$, or $H_{5} = 2K_2 \oplus F_3$, respectively.  This completes Case 1.

\medskip
\noindent
\textbf{Case 2}: $G[A]$ and $G[B]$ contain $P_3$.
\medskip

\noindent
Since $G$ is a cograph, it has no induced $C_5$. Together with the fact that a threshold graph is a $(C_4, P_4, 2K_2)$-free graph it suffices to consider the following cases. 

\medskip
\noindent
\textbf{Subcase 2.1}: $G[A]$ contains  $C_4$.
\medskip

\noindent
Then $G$ contains $H_1 = C_4 \oplus P_3$.

\medskip
\noindent
\textbf{Subcase 2.2}: $G[A]$ contains $2K_2$. 
\medskip

\noindent
By Lemma \ref{second}, $G[A]$ contains $Q_1$ or $Q_2$. It follows that $G$ contains $H_4 = P_3 \oplus Q_1$ or $H_2 = P_3 \oplus Q_2$, respectively.

\medskip
\noindent
\textbf{Subcase 2.3}: $G[A]$ and $G[B]$ are threshold graphs.
\medskip

\noindent
It follows by Lemma \ref{basiclemma1} that $G$ is partitionable.  This completes Case 2 and the proof of Claim~\ref{c1}. 
\end{proof}

\begin{claim}\label{c2}
If $G$ has a universal vertex $v$ such that $G' = G \sm v$ is disconnected then $G$ contains one of the graphs $H_9, H_{10}, H_{11}, H_{12}$.
\end{claim}

\begin{proof}[Proof of Claim~\ref{c2}]
Let $r\geq 2$ be an integer and $\mathcal{G} = \{G_1, \dots, G_r\}$ be the set of components of $G'$. By the minimality of $G$ for every graph $G_i \in \mathcal{G}$ the graphs $G_i$ and $G'_i = v \oplus G_i$ are partitionable. We claim that there exists a graph $T \in \mathcal{G}$ that is $(1, 2)$-partitionable but not monopolar.

To see this consider a graph $K \in \mathcal{G}$. If every partition of $K$ into $k$ disjoint cliques and $l$ independent sets has $\min(k, l) \geq 2$ then $K' = v \oplus K$ is in-partitionable. So we may assume that each $G_i \in \mathcal{G}$ is either $(1, 2)$-partitionable or monopolar. But If every $G_i \in \mathcal{G}$ is monopolar then $G'_i$ admits a partition where $v$ is in the bipartite part. Hence, as the $G_i$'s are disjoint, $G$ also admits a partition where $v$ is  again in the bipartite part. 

From now on, let $G_j \in \mathcal{G}$  be a graph that is $(1, 2)$-partitionable but not monopolar for some $j \in \{1, \dots r\}$. By Theorem \ref{monopolarthm1}  and Remark \ref{monopolarcoro}, $G_j$ contains one of the graphs $J_1, J_2, J_3, J_4$. For contradiction suppose  there exists no $p \not= j$ such that $G_p$ contains $K_3$. Let $C(G_j)$ and $S(G_j)$ denote the partition of $G_j$ into a clique and a bipartite graph, respectively. Then $V = A \cup B$ where $A = v \cup C(G_j)$, and $B=S(G_j) \cup \bigcup_{p \not=j} G_p $ is a partition of $V$ where $G[A]$ is $P_3$-free and $G[B]$ is bipartite, a contradiction. 

We conclude that $G$ contains $H_9 = v \oplus (K_3 \cup J_1)$, $H_{10} = v \oplus (K_3 \cup J_2)$, $H_{11} = v \oplus (K_3 \cup J_3)$ or $H_{12} = v \oplus (K_3 \cup J_4)$. 
\end{proof}

\begin{claim}\label{c3}
If $G$ has a universal vertex $v$ such that $G' = G \sm v$ is connected then $G$ contains one of the graphs $H_1, H_2, H_4, H_5, H_{13}, H_{14}, H_{15}, H_{17}$.
\end{claim}

\noindent
\textit{Proof of Claim~\ref{c3}}. By the minimality of $G$, $G'$ is partitionable. In particular, $G'$ is neither monopolar nor $(1,2)$-partitionable, otherwise $G = G' \oplus v$ is partitionable. Hence, by Theorem \ref{almostsplitthm}, $G'$ contains one of the graphs $R_1,  \dots, R_8$.  It follows that $G$ contains $H_1 = v \oplus R_1$, $H_5 = v \oplus R_2$, $H_4 = v \oplus R_3$, $H_2 = v \oplus R_4$, $H_{17} = v \oplus R_5$, $H_{13} = v \oplus R_6$, $H_{14} = v \oplus R_7$, or $H_{15} = v \oplus R_8$. This completes the proof of Claim~\ref{c3} and Theorem \ref{mainthm}.
  \end{proof}

\section{Further Work}\label{fw}

Chudnovsky described in a series of papers ~\cite{chud1, chud2} a decomposition theorem of bull-free graphs. In \cite{chud2} the basic graph class $\mathcal{T}_1$ is the set of bull-free graphs $G$ that admit a partition $\{A, B\}$ of $V(G)$ such that $A$ induces a triangle-free graph and $B$ induces a disjoint union of cliques together with some adjacency constraints between $A$ and $B$. Unfortunately our reduction for the bull-free case does not satisfy these adjacency constraints. Hence the recognition of the class $\mathcal{T}_1$ remains an open problem. 

A graph is Meyniel if every odd cycle of length at least $5$ contains at least two chords. Meyniel graphs are between chordal and perfect graphs. Since our partition problem is tractable in the former case but \NP-complete in the latter, it would be of interest to narrow this complexity gap by focusing on Meyniel graphs. 

A possible extension of our result on cographs is the following. Given a finite sequence $(H_1, \dots ,H_k)$ of cographs, can we compute the finite set $F$ of
cographs such that for every cograph $G$, the vertices of G can be partitioned into $V_1, \dots,V_k$ such that $G[V_i]$ is $H_i$-free if and only if
$G$ is $F$-free? By Damaschke's result \cite{damaschke} we know that such a finite set $F$ of forbidden induced subgraphs exists. It would be enough to prove a recursive bound on the size of the graphs in $F$. For $k=2$, $H_1=K_3$ and $H_2=P_3$ we described the set $F$ in Section \ref{cographsection}.

 Another more general problem to consider is the following. Let $G = (V, E)$ be a graph, and let $\mathcal{F}$ and $\mathcal{Q}$ be fixed additive induced hereditary properties. The problem of deciding whether $V$ has a partition $\{A, B\}$ such that $G[A]$ is $\mathcal{F}$-free and $G[B]$ is $\mathcal{Q}$-free is $\NP$-complete \cite{alastair}. What is the complexity of this problem when restricted to special graph classes?

\section*{Acknowledgments}

Thanks are due to Matthew Johnson for carefully reviewing the paper. We thank the referees for helpful suggestions and pointing out an error in the previous manuscript.


\end{document}